\documentclass[iicol,sn-mathphys-num]{sn-jnl}
\usepackage[T1]{fontenc}

\usepackage{amsmath,amssymb,amsfonts}%
\usepackage{amsthm}
\usepackage{mathrsfs}
\usepackage{mathtools}
\usepackage{etoolbox} 
\usepackage{dsfont}

\usepackage{xcolor}
\usepackage{url}
\usepackage{wrapfig}
\usepackage{multirow}
\usepackage{booktabs}
\usepackage{graphicx}

\usepackage{algorithmic}
\usepackage{algorithm}

\usepackage{enumitem}

\def\nN{\mathcal{N}}
\def\fJ{J}

\def\fF{\mathcal{F}}

\def\R{\mathbb{R}}

\def\N{\mathbb{N}}

\def\G{\mathcal{G}}
\def\K{\mathbf{K}}
\def\L{\mathbf{L}}
\def\cK{\mathbf{K}}

\def\S{\mathbf{S}}
\def\sS{\mathbf{S}}
\usepackage[labelfont=bf,justification=justified,singlelinecheck=false]{caption}

\newcommand{\Prox}[2]{\mathsf{Prox}_{#1}\left(#2\right)}
\newcommand{\Prb}[1]{\mathbb{P}\left(#1\right)}

\newcommand{\eE}{\mathbb{E}}

\DeclareMathOperator*{\argmin}{arg\,min}
\DeclareFontFamily{U}{matha}{\hyphenchar\font45}
\DeclareFontShape{U}{matha}{m}{n}{
<-6> matha5 <6-7> matha6 <7-8> matha7
<8-9> matha8 <9-10> matha9
<10-12> matha10 <12-> matha12
}{}
\DeclareSymbolFont{matha}{U}{matha}{m}{n}
\DeclareFontFamily{U}{mathx}{\hyphenchar\font45}
\DeclareFontShape{U}{mathx}{m}{n}{
<-6> mathx5 <6-7> mathx6 <7-8> mathx7
<8-9> mathx8 <9-10> mathx9
<10-12> mathx10 <12-> mathx12
}{}
\DeclareSymbolFont{mathx}{U}{mathx}{m}{n}
\DeclareMathDelimiter{\vvvert} {0}{matha}{"7E}{mathx}{"17}
\DeclarePairedDelimiterX{\normiii}[1]
{\vvvert}
{\vvvert}
{\ifblank{#1}{\:\cdot\:}{#1}}

\theoremstyle{thmstyleone}
\newtheorem{theorem}{Theorem}
\newtheorem{proposition}[theorem]{Proposition}
\newtheorem{corollary}{Corollary}
\theoremstyle{thmstyletwo}

\newtheorem{remark}{Remark}
\newtheorem{lemma}{Lemma}
\theoremstyle{thmstylethree}
\newtheorem{definition}{Definition}
\newtheorem{assumption}{Assumption}

\begin{document}

\title[Equivariant Denoisers for Plug and Play Image Restoration]{Equivariant Denoisers for Plug and Play Image Restoration}
\author[1,2]{\fnm{Marien} \sur{Renaud}}\email{marien.renaud@math.u-bordeaux.fr}
\author[1]{\fnm{Eliot} \sur{Guez}}
\author[1]{\fnm{Arthur} \sur{Leclaire}}
\author[2]{\fnm{Nicolas} \sur{Papadakis}}

\affil[1]{\orgdiv{LTCI}, \orgname{T\'el\'ecom Paris}, IP Paris}
\affil[2]{\orgname{Univ. Bordeaux, CNRS, INRIA, Bordeaux INP, IMB, UMR 5251}, \postcode{F-33400} \city{Talence}, \country{France}}

\abstract{
One key ingredient of image restoration is to define a realistic prior on clean images to complete the missing information in the observation.
State-of-the-art restoration methods rely on a neural network to encode this prior. 
Typical image distributions are invariant to some set of transformations, such as rotations or flips.
However, most deep architectures are not designed to represent an invariant image distribution.
Recent works have proposed to overcome this difficulty by including equivariance properties within a Plug-and-Play paradigm.
In this work, we propose two unified frameworks named Equivariant Regularization by Denoising (ERED) and Equivariant Plug-and-Play (EPnP) based on equivariant denoisers and stochastic optimization. We analyze the convergence of the proposed algorithms and discuss their practical benefit.}
\maketitle

\section{Introduction}
Image restoration aims at recovering a proper image $x \in \R^d$ from a degraded observation ${y \in \R^m}$. 
A model to obtain $y$ from $x$ can be defined as $y \sim \mathcal{N}(\mathcal{A}(x))$,
where $\mathcal{A}: \R^d \to \R^m$ is a deterministic operation on $x$ and $\mathcal{N}$ is a noise distribution that corresponds to a known physical degradation. Typically, linear degradations with Gaussian noise can be written as $y = A x + n$, with $A \in \R^{m \times d}$ and $n \sim \nN(0, \sigma_y^2 I_m)$.

The restoration task can then be reformulated as a variational problem
\begin{align}\label{eq:ideal_opt_pblm}
    \argmin_{x \in \R^d} \fF(x) = f(x) + \lambda r(x),
\end{align}
where the data-fidelity term $f = - \log{p(y|\cdot)}$ is the log-likelihood representing the degradation, while $r = - \log p$ is the regularization term encoding the prior distribution $p$, i.e. the model for clean images. 
The choice of the regularization $r$ is crucial to complete the missing information in the observation $y$. For example, classical explicit regularizations favor piecewise constant images~\cite{rudin1992nonlinear} or wavelet sparsity~\cite{mallat1999wavelet}.

\textbf{\textit{Implicit regularization with RED and PnP}}
To regularize the problem~\eqref{eq:ideal_opt_pblm}, state-of-the-art generic methods rely on the use of an external denoiser within the REgularization by Denoising (RED)~\cite{romano2017little} or Plug-and-Play (PnP)~\cite{venkatakrishnan2013plug} frameworks. These methods build upon an off-the-shelf denoiser that is plugged within established first order optimization algorithms ({\em e.g.} Proximal Gradient Descent) in place of explicit (RED) or implicit (PnP) gradient descent steps over the regularization function $r$.

This denoiser $D_{\sigma}$, parametrized by a noise level $\sigma > 0$, is generally a pretrained deep neural network~\cite{zhang2021plug}. Using a deep denoiser to encode the regularization has indeed proved to be state-of-the-art in many real-world applications including SAR despeckling~\cite{deledalle2017mulog}, phase retrieval~\cite{Metzler2018},  MRI acquisition~\cite{liu2020rare,iskender2023red}, tomographic reconstruction~\cite{majee2021multi}, Rician noise removal~\cite{wei2023nonconvex}, traffic-data completion~\cite{chen2024low} or structure-texture decomposition~\cite{guennec2025joint}.

\textbf{\textit{Constrained denoisers for convergent schemes}}
Obtaining convergence result for PnP or RED is an active field of research to have a theoretical guarantee that the restoration algorithm is stable and reliable.
Depending of the optimization scheme, convergence result can be derived by constraining the denoiser being bounded~\cite{chan2016plug, gavaskar2019proof}
having a symetric Jacobian~\cite{romano2017little,cohen2021regularization}, being averaged~\cite{Sun_2019}, being firmly nonexpansive~\cite{terris2020building,sun2021scalable}, being simply nonexpansive~\cite{reehorst2018regularization,ryu2019plug,liu2021recovery},
being pseudo-contractive~\cite{wei2024learning},
or being a gradient field~\cite{hurault2022gradient,hurault2024convergent}. The more recent conditions~\cite{hurault2022gradient,wei2024learning} can be ensured for deep denoiser with a slight degradation of their denoising perform.

\textbf{\textit{Image invariances}} Image distributions can be expected to be invariant to some transformations such as rotations or flips. 
Therefore, data augmentations methods~\cite{lecun2002gradient,krizhevsky2012imagenet} relies on these invariances. However, it is not clear that existing trained denoisers encode invariances on the underlying prior distribution. 
Recent works develop restoration methods that take advantage of equivariance to improve restoration with
rotation~\cite{fu2023rotationequivariantproximaloperator,terris2024equivariant} or roto-translation~\cite{celledoni2021equivariant}. Equivariance have also been used to developed
unsupervised image restoration~\cite{chen2021equivariant}, bootstrapping for uncertainty quantification~\cite{tachella2023equivariantbootstrappinguncertaintyquantification} or
improve the quality of a pretrained score for diffusion model~\cite{mbakam2024empirical}.

\textbf{\textit{Stochastic versions of RED and PnP}} Another way to increase the performance of image restoration methods is to rely on stochastic versions of optimization algorithms, such as the Stochastic Proximal Gradient descent algorithm~\cite{nitanda2014stochastic,xiao2014proximal}.
Different strategy of stochastic RED or PnP have thus been proposed, such as the equivariant methods~\cite{terris2024equivariant,renaud2025equivariant} that correspond to stochastic optimization algorithms.
Previously, the authors of~\cite{laumont2023maximum} proposed to add Gaussian noise in the algorithm to better explore the landscape. In~\cite{renaud2024plug,renaud2024convergenceanalysisproximalstochastic}  noise is added to the image before denoising in order to ensure that the denoiser is applied on its training domain during inference. In the same line of work, the authors of~\cite{hu2024stochasticdeeprestorationpriors} use a random degradation-restoration operation to regularize.
Other works  resulting in stochastic algorithms compute mini-batch approximation of the regularization~\cite{sun2019block} or the data-fidelity~\cite{tang2020fast} in order to reduce the computational cost.

\textbf{\textit{Links with sampling algorithms}}
Another strategy to solve inverse problems is to sample the posterior distribution $p(x|y)$ instead of solving the optimization problem~\eqref{eq:ideal_opt_pblm}. For sampling, deep denoisers are also used to encode the prior knowledge on the distribution of clean images within Gibbs sampling~\cite{coeurdoux2024plug,sun2024provable,Faye_IEEE_Trans_IP_2025}, Langevin dynamics~\cite{Laumont_2022,renaud2023plug,habring2025diffusion,renaud2025stability} or diffusion model~\cite{songscore,kadkhodaie2020solving,xia2023diffir} frameworks. Recent works show that diffusion models can be reformulated as relaxed noising-denoising algorithm~\cite{leclaire2025backward}. Thus the algorithm gap between sampling~\cite{chung2022diffusion} and stochastic optimization algorithms~\cite{renaud2024plug} is subtle. This fine gap has been previously observed in the series of works~\cite{Laumont_2022,laumont2023maximum}, which relate the correspondence between the unadjusted Langevin algorithm and a stochastic gradient algorithm to a different weighting of the noise term. In this paper, we will recover this correspondence: indeed, we prove that the proposed SnoPnP Algorithm~\ref{alg:SnoPnP} aims for optimizing (see Proposition~\ref{prop:snopnp_convergence_critical_point}), whereas its PnP-PSGLA counterpart~\cite{renaud2025stability} (which includes a different normalization of the noise component) aims for sampling.

\textbf{\textit{Nonconvex stochastic optimization}}
In terms of algorithmic convergence, RED and PnP generally correspond to first-order algorithms applied to the optimization of  nonconvex problems~\eqref{eq:ideal_opt_pblm}. 
When coming to the stochastic optimization of nonconvex problems, an extensive literature~\cite{atchade2014stochastic,konevcny2015mini,ghadimi2016mini,allen2018katyusha,davis2019proximally,j2016proximal} focuses on Stochastic Proximal Gradient descent motivated by the application of learning deep neural works with a convex constraint on the weights.
These works have been extended by including variance reduction~\cite{li2018simple,xu2023momentum,lu2024variance} or momentum~\cite{gao2024non} mechanisms or generalized to Bregman divergences~\cite{ding2023nonconvex}.
More recently some works have extended the previous guarantees for weakly convex regularizations~\cite{li2022unified,renaud2024convergenceanalysisproximalstochastic}.\\

However, for equivariant image restoration, the noise is added intentionally to improve the performance, discarding previous proof strategies. Namely, for the Perturbed Proximal Gradient descent algorithm that we introduce in section~\ref{sec:epnp}, the proximal evaluations are randomly perturbed. To the best of our knowledge, that kind of optimization algorithm has never been analyzed. The optimization results presented here for these algorithms are new, and can be applied beyond the context of image restoration.

\textbf{\textit{Contributions and outline}}
In this paper, we propose a unified formalism for the ERED and PnP frameworks, which generalizes the so-called equivariant PnP~\cite{terris2024equivariant}, equivariant RED~\cite{renaud2025equivariant}  as well as other stochastic versions of related algorithms~\cite{hu2024stochasticdeeprestorationpriors,renaud2024plug}.

In section~\ref{sec:back}, we start by introducing the main RED and PnP concepts that will be used all along the document.
In section~\ref{sec:ered}, we recall the $\pi$-Equivariant Regularization by Denoising (ERED) of~\cite{renaud2025equivariant} and give theoretical insights on the convergence guarantees and the critical points behavior of ERED (Algorithm~\ref{alg:ERED}).

Then we extend ERED to PnP in section~\ref{sec:epnp}, by introducing the Equivariant Plug-and-Play algorithm (EPnP, Algorithm~\ref{alg:EPnP}), for which we provide convergence guarantees. 
In section~\ref{sec:snopnp}, we refine these guarantees in the particular case of the Stochastic Denoising Plug-and-Play (SnoPnP, Algorithm~\ref{alg:SnoPnP}).

We finally provide in Section~\ref{sec:experimental} numerical experiments and comparisons for image restoration tasks, and discuss the practical benefits of equivariant approaches.

With respect to the short paper~\cite{renaud2025equivariant}, we propose a new extension of equivariant algorithms in the PnP framework (EPnP and SnoPnP) and provide the associated convergence results. Moreover, advanced numerical experiments and comparisons have been conducted to validate the ERED, EPnP and SnoPnP algorithms.

\section{Background}\label{sec:back}
In order to solve problem~\eqref{eq:ideal_opt_pblm}, when $f$ and $r$ are differentiable, one can use a gradient descent algorithm. However, the gradient of $r$, i.e. the score of the prior distribution $p$ of clean images $s := \nabla \log p = - \nabla r$ is unknown. The authors of~\cite{romano2017little} proposed to make the following approximation 
\begin{align}
    \nabla r = - \nabla \log p \approx \nabla r_{\sigma} := - \nabla \log p_{\sigma},
\end{align} 
where $p_{\sigma} = \nN_{\sigma} \ast p$ is the convolution of $p$ with the Gaussian $\nN_{\sigma}$ with $0$-mean and $\sigma^2 I_d$-covariance matrix. This is motivated by the Tweedie formula~\cite{efron2011tweedie}, which makes $\nabla \log p_{\sigma}$ tractable
\begin{align}\label{eq:tweedie_formula}
    -\nabla \log p_{\sigma}(x) = \frac{1}{\sigma^2} \left( x - D_{\sigma}^\ast(x) \right),
\end{align}
where $D_{\sigma}^\star$ is the Minimum Mean Square Error (MMSE) denoiser defined by $D_{\sigma}^\star(\tilde x) := \eE [x | \tilde x ] = \int_{\R^d}{x p(x | \tilde x) dx}$,
for $\tilde x = x + \epsilon \text{ with } x \sim p(x), \epsilon \sim \nN(0, \sigma^2 I_d)$.
This approximation leads to the REgularization by Denoising (RED) iterations:
\begin{align}
    x_{k+1} &=  x_k - \delta \nabla f(x_k) - \delta \frac{\lambda}{\sigma^2} \left( x_k - D_{\sigma}(x_k)\right),
\end{align}
where $D_\sigma$ is a denoiser that is designed to approximate $D_{\sigma}^{\star}$.

Another approach, named Plug-and-Play (PnP)~\cite{venkatakrishnan2013plug}, is to consider a forward-backward algorithm to solve problem~\eqref{eq:ideal_opt_pblm}, that can be written as
\begin{align}\label{eq:forward_backward}
    x_{k+1} = \Prox{\delta r}{x_k - \frac{\delta}{\lambda} \nabla f(x_k)},
\end{align}
where the proximal operator of the regularization $r = - \log p$ is defined for $x \in \R^d$ by
\begin{align}\label{eq:interpretation_pnp}
    \Prox{-\delta \log p}{x} = \argmin_{z \in \R^d} \frac{1}{2 \delta}\|x - z \|^2 - \log p(z),
\end{align}
with $\delta > 0$ the step-size and $p$ the prior distribution. One can observe that the optimization problem in equation~\eqref{eq:interpretation_pnp} amounts to computing the maximum of the posterior distribution of the denoised image $x \in \R^d$ with respect to the prior $p$ and an additive Gaussian noise of standard deviation $\sqrt{\delta}$. Therefore, several works~\cite{sreehari2016plug,hurault2024convergent} proposed to replace $\mathsf{Prox}_{-\delta \log p}$ in equation~\eqref{eq:forward_backward} with a pretrained denoiser $D_{\sigma}$ leading to the Plug-and-Play iterations:
\begin{align}\label{eq:pnp}
    x_{k+1} = D_{\sigma}\left(x_k - \frac{1}{\lambda} \nabla f(x_k)\right).
\end{align}
Note that, in general, the denoiser $D_{\sigma}$ is not guaranteed to be the proximal operator of some regularization and $\sigma > 0$ is not a step size. Under the interpretation of equation~\eqref{eq:interpretation_pnp}, we should have $\sigma = \sqrt{\delta}$. We choose to keep $\sigma > 0$ disconnected with $\delta$ in order to be aligned with the PnP literature~\cite{zhang2021plug,hurault2022gradient}.
Moreover, in the iterations~\eqref{eq:pnp} the two parameters $\lambda$ and $\delta$ are redundant, therefore we choose to keep $\lambda > 0$ and fix $\delta = 1$ with respect to~\eqref{eq:forward_backward}.

The performance of RED and PnP can be improved by slightly modifying the algorithm with stochastic schemes~\cite{terris2024equivariant,renaud2024plug}, which incorporate invariance properties and enhance details in the restoration. In the next sections, we propose a unified framework to generalize these approaches.

\section{\texorpdfstring{$\pi$}{pi}-Equivariant Regularization by Denoising}\label{sec:ered}
We now  study ERED,  the equivariant extension of RED proposed in~\cite{renaud2025equivariant}, that is deduced from a notion of invariance, named $\pi$-equivariance, on the underlying prior. 
We first introduce the general notions of $\pi$-equivariance in section~\eqref{ssec:pieq} and corresponding equivariant regularizations in section~\ref{ssec:eqreg}. Then we exemplify several case corresponding to existing works in section~\ref{sec:equivariant_ex}. Section~\ref{ssec:ERED} then describes the ERED algorithm, while the remaining sections are dedicated to its convergence and critical point analyses. For the sack of completeness, we recall in appendices all the proofs provided in~\cite{renaud2025equivariant}.

\subsection{\texorpdfstring{$\pi$}{pi}-equivariant image distributions}\label{ssec:pieq}
As shown in~\cite{lenc2015understanding}, natural images densities tend to be invariant to some set of transformations such as rotations or flips. To formalize these properties, we  define the key notions of invariance and $\pi$-equivariance.

\begin{definition}[Invariance]\label{def:invariance}
We denote by $g : \R^d \to \R^d$ a differentiable transformation of $\R^d$ and by $\G$ a measurable set of transformations of $\R^d$.
    A density $p$ on $\R^d$ is said to be invariant to a set of transformations $\G$ if $\forall g \in \G$, $p = p \circ g$ a.e.
\end{definition}

\begin{definition}[$\pi$-equivariance]\label{def:equivariance}
We denote by $G \sim \pi$ a random variable of law $\pi$ on $\G$.
    A density $p$ on $\R^d$ is said to be $\pi$-equivariant if 
    $\eE_{G \sim \pi}[|\log(p \circ G)|] < \infty$
    and 
    $\log p = \eE_{G \sim \pi}\left[ \log(p \circ G) \right]$.
\end{definition}

Definition~\ref{def:equivariance} relaxes the notion of invariance for a density in the following sense. If a density $p$ is invariant to each $g \in \G$, $p$ is $\pi$-equivariant, whatever the distribution $\pi$ on $\G$. 

\begin{remark}\label{remark:score_invariance}
For $p \in \mathcal{C}^1(\R^d, \R_{+}^\ast)$ and $g \in \mathcal{C}^1 (\R^d, \R^d)$, we have
\begin{align}\nonumber
    \nabla \log (p \circ g)(x)
    &= \frac{\nabla (p \circ g)(x)}{(p \circ g)(x)} = \frac{\fJ_g^T(x)\nabla p(g(x))}{(p \circ g)(x)}\\ &= \fJ_g^T(x) (\nabla \log p)(g(x)),\label{eq:score_composition}
\end{align}
with $x \in \R^d$. Thus, if $p$ is $\pi$-equivariant, then $s=-\nabla \log p$, the score of $p$,  verifies the identity $s = \eE_{G \sim \pi}\left( \fJ_G^T (s \circ G) \right)$.

\end{remark}

In the context of equivariant transforms, Remark~\ref{remark:score_invariance} suggests to apply the Jacobian of the transformation $g$ instead of the inverse of the transformation $g^{-1}$ as it is done in existing  works~\cite{mbakam2024empirical,tachella2023equivariantbootstrappinguncertaintyquantification,terris2024equivariant,herbreteau2024normalization}. Here the score~\eqref{eq:score_composition} can be computed for any general differentiable transformation $g$, even if $g^{-1}$ does not exist.

\subsection{Equivariant regularization}\label{ssec:eqreg}
In order to encode the desired equivariance property in the regularization $r$, \cite{renaud2025equivariant} introduces the \textbf{Equivariant REgularization by Denoising (ERED)} $ r_{\sigma}^{\pi}$ and the associated score $s_{\sigma}^{\pi}$ respectively defined by
\begin{align}
    r_{\sigma}^{\pi}(x) &:= -\eE_{G \sim \pi} \left( \log (p_{\sigma} \circ G)(x) \right) \label{eq:eq_reg} \\
    s_{\sigma}^{\pi}(x) &:= -\eE_{G \sim \pi} \left(\fJ_G^T(x) (\nabla \log p_{\sigma})(G(x)) \right).\label{eq:equivariant score}
\end{align}

Note that under regularity assumptions on $\pi$ and $p_{\sigma} \circ G$ (e.g. $\mathcal{G}$ finite and $p_{\sigma} \circ G$ differentiable), we get $s_{\sigma}^{\pi} = \nabla r_{\sigma}^{\pi}$.
Thanks to the Tweedie formula~\eqref{eq:tweedie_formula}, $s_{\sigma}^{\pi}$ can be computed with an MMSE denoiser
\begin{align}\label{eq:equiv_score}
    s_{\sigma}^{\pi}(x) &= \eE_{G \sim \pi} \left(\frac{1}{\sigma^2}  \fJ_G^T(x) \left(G(x) - D_{\sigma}^\ast(G(x)) \right) \right).
\end{align}
For a given denoiser $D_{\sigma}$, e.g. a supervised neural network, we can thus introduce the \textit{equivariant denoiser} $\tilde D_{\sigma}$ defined~by
\begin{align}\label{eq:equiv_denoiser}
    \tilde D_{\sigma}(x) = \eE_{G \sim \pi} \left[ \fJ_{G}^T(x) D_{\sigma} \left( G (x) \right) \right].
\end{align}

Since the exact MMSE denoiser $D_{\sigma}^\ast$ is not tractable, we make the following approximation:
\begin{align*}
    s_{\sigma}^{\pi}(x) &= \frac{1}{\sigma^2} \left( \eE_{\pi} \left[\fJ_G^T(x) G(x) \right] - \tilde D_{\sigma}^\ast(x)  \right)\\& \approx \frac{1}{\sigma^2} \left( \eE_{\pi} \left[\fJ_{G}^T(x) G(x) \right] - \tilde D_{\sigma}(x) \right).
\end{align*}

No special structure is required on $\G$ or $\pi$, thus making the  ERED framework a generic construction. 
However, in order to ensure that the ERED $ r_{\sigma}^{\pi}$ is indeed $\pi$-equivariant (Definition~\ref{def:equivariance}), more structure is required on $\G$ and $\pi$. 
In Proposition~\ref{prop:haar}, a sufficient condition on $\G$ and $\pi$ is provided for this property to hold.
Before stating this result, let us recall that a compact Hausdorff topological group admits a unique right-invariant probability measure $\pi$, called Haar measure, that satisfies, for any integrable function $\varphi : \G \to \R$ and for any $g \in \G$, $\int_{\G} \varphi(g(x)) d\pi(x) = \int_{\G} \varphi(x) d\pi(x)$~\cite{haar1933massbegriff,neumann1935haarschen}. For a finite group $\G$, the Haar measure is the counting measure.

\begin{proposition}\label{prop:haar}
    If $\G$ is a compact Hausdorff topological group and $\pi$ the associated right-invariant Haar measure, then $r_{\sigma}^{\pi}$ is $\pi$-equivariant.
\end{proposition}

Proposition~\ref{prop:haar} is proved in Appendix~\ref{sec:proof_haar}
Let us discuss the hypothesis on the set of transformations $\G$.
First, $\G$ needs to be a group, i.e. $\forall g, g' \in \G, g^{-1} \circ g' \in \G$, to ensure that the composition $G \circ G'$ is still in $\G$ and that any transformation is invertible.
Moreover, $\G$ needs to be a Hausdorff space (i.e. $\forall g, g' \in G$, there exist two neighbourhoods $U$ and $V$ of  $x$ and $y$, such that $U \cap V = \emptyset$) and is required to be compact to ensure that $\pi$ is a probability measure.
These hypotheses on $\G$ are general and cover in particular any finite discrete $\G$.

\subsection{Examples of equivariant scores}\label{sec:equivariant_ex}
The equivariant formulation generalizes recent works in the literature that we recall here. Moreover, we present new examples of equivariant scores that are included in our general framework.

\textbf{Finite set of isometries}
The authors of \cite{terris2024equivariant} proposed an equivariant version of PnP for a finite set of linear isometric transformations $\G$ with the uniform distribution on $\G$, i.e. $\forall g \in \G, \pi(g) = \frac{1}{|\G|}$. Since $g$ is a linear isometry, $\fJ_g^T(x) = g^{-1}$.
In this case, the $\pi$-equivariant denoiser~\eqref{eq:equiv_denoiser} and  regularization~\eqref{eq:equiv_score} are respectively defined by
\begin{align}
    \tilde D_{\sigma}(x) &= \frac{1}{|\G|} \sum_{g \in \G}{ g^{-1}\left[ D_{\sigma} \left( g(x) \right) \right] },\\
     s_{\sigma}^{\pi}(x) &\approx \frac{1}{\sigma^2} \left(x - \tilde D_{\sigma}(x) \right).
\end{align}

\textbf{Infinite set of isometries}
To generalize the previous example, we can propose an infinite set $\G$ of isometries, e.g. sub-pixel rotations for which $\pi$ can be seen as the angular distribution.
In this case, $\forall g \in \G, \forall x \in \R^d,\  \fJ_g^T(x) = g^{-1}$ and then
\begin{align}
    \tilde D_{\sigma}(x) &= \eE_{G \sim \pi}\left( G^{-1}\left[ D_{\sigma} \left( G(x) \right) \right] \right),\\s_{\sigma}^{\pi}(x) &\approx \frac{1}{\sigma^2} \left(x - \tilde D_{\sigma}(x) \right).
\end{align}

\textbf{Noising-denoising} The stochastic denoising regularization (SNORE) proposed in~\cite{renaud2024plug} consists in noising the image before denoising it. It can be interpreted as a $\pi$-equivariant denoiser for the set of translations $g_{z}(x) = x + \sigma z$, for $x, z \in \R^d$ and $\sigma > 0$ the noise level of $D_{\sigma}$, with the multivariable distribution, i.e. ${\forall z \in \R^d}$, $\pi(g_{z}) = \nN(z; 0, I_d)$. The Jacobian of the translation is $\fJ_{g_z}(x) = I_d$. With this set of transformations, the $\pi$-equivariant  denoiser~\eqref{eq:equiv_denoiser} and regularization~\eqref{eq:equiv_score} can be expressed as 
\begin{align}
    \tilde D_{\sigma}(x) &= \eE_{z \sim \nN(0, \sigma I_d)} \left[ D_{\sigma} \left( x + \sigma z \right) \right] \label{eq:noising_denoising_denoiser} \\
    s_{\sigma}^{\pi}(x) &\approx \frac{1}{\sigma^2} \left( \eE_{\pi} \left[x + \sigma z \right] - \tilde D_{\sigma}(x) \right)\nonumber \\&= \frac{1}{\sigma^2} \left( x - \tilde D_{\sigma}(x) \right).
\end{align}

\subsection{ERED algorithm}\label{ssec:ERED}

In this section, we define a generic $\pi$-equivariant PnP algorithm and demonstrate its convergence.
Instead of solving Problem~\eqref{eq:ideal_opt_pblm}, we will tackle 
\begin{align}\label{eq:equivariant_opt_pblm}
    \argmin_{x \in \R^d} \mathcal{F}_{\sigma}^{\pi}(x) := f(x) + \lambda r_{\sigma}^{\pi}(x),
\end{align}
with the equivariant regularization $r_{\sigma}^{\pi}$ defined in relation~\eqref{eq:eq_reg}.
We now introduce the equivariant Regularization by Denoising (ERED) algorithm (Algorithm~\ref{alg:ERED}), which is a biased stochastic gradient descent to solve Problem~\eqref{eq:equivariant_opt_pblm}.

\begin{algorithm}
\caption{ERED}\label{alg:ERED}
\begin{algorithmic}[1]
\STATE \textbf{Parameters:} $x_0 \in \R^d$, $\sigma > 0$, $\lambda > 0$, $\delta > 0$, $N \in \N$
\STATE \textbf{Input:} degraded image $y$
\STATE \textbf{Output:} restored image $x_{N}$
\FOR{$k = 0, 1, \dots, N-1$}
    \STATE Sample $G \sim \pi$
    \STATE $x_{k+1} = x_k - \delta \nabla f(x_k) - \frac{\delta \lambda}{\sigma^2} \fJ_{G}^T(x_k) \left( G(x_k) - D_{\sigma}(G(x_k))\right)$
\ENDFOR
\end{algorithmic}
\end{algorithm}

Note that Algorithm~\ref{alg:ERED} is a generalization of RED, of the previous equivariant RED proposed in~\cite{terris2024equivariant} and of SNORE~\cite{renaud2024plug}.

\begin{remark}
    If the denoiser $D_{\sigma}$ is $L$-Lipschitz then the equivariant denoiser $\tilde D_{\sigma}$ defined in relation~\eqref{eq:equiv_denoiser} is $L$-Lipschitz. 
    Previous works have shown the link between the Lipschitz constant of the denoiser and the convergence of deterministic PnP algorithms~\cite{ryu2019plug,hurault2022proximal,wei2024learning}. However, no clue indicates that this property still holds for stochastic PnP algorithms, such as Algorithm~\ref{alg:ERED}.
\end{remark}

\subsection{ERED unbiased convergence analysis}

In this section, we prove the convergence of the ERED (Algorithm~\ref{alg:ERED}) run with the exact MMSE denoiser~$D_{\sigma}^\star$. With this denoiser, thanks to Tweedie formula, the iterations are computed by
\begin{align}\nonumber
    x_{k+1} = &\; x_k - \delta_k \nabla f(x_k)\\&- \lambda \delta_k \fJ_{G}^T(x_k) \nabla \log p_{\sigma}(G (x_k)),\label{eq:theoretical_process}
\end{align}
with $G \sim \pi$ and $(\delta_k)_{k \in \N} \in {\left(\R^+\right)}^{\N}$ a non-increasing sequence of step-sizes.

\begin{assumption}
    \itshape \textbf{(a)} \label{ass:step_size_decreas} The step-size decreases to zero but not too fast: $\sum_{k = 0}^{+\infty}{\delta_k} = + \infty$ and $\sum_{k = 0}^{+\infty}{\delta_k^2} < + \infty$.
    
    \itshape \textbf{(b)} \label{ass:data_fidelity_reg} The data-fidelity term $f : x \in \R^d \mapsto f(x) \in \R$ is $\mathcal{C}^{\infty}$.
    
    \itshape \textbf{(c)} \label{ass:prior_score_approx} The noisy prior score is sub-polynomial, i.e. there exist $B\in \R^+$, $\beta \in \R$ and $n_1 \in \N$ such that $\forall \sigma > 0$, $\forall x \in \R^d$, $\|\nabla \log p_{\sigma}(x)\| \le B \sigma^{\beta} (1 + \|x\|^{n_1})$.
\end{assumption}

Assumption~\ref{ass:step_size_decreas}(a) is standard in stochastic gradient descent analysis. It suggests a choice of the step-size rule to ensure convergence, for instance $\delta_k = \frac{\delta}{k^\alpha}$ with $\alpha \in ]\frac{1}{2},1]$.
Assumption~\ref{ass:data_fidelity_reg}(b) is typically verified for a linear degradation with additive Gaussian noise, i.e. $f(x) = \frac{1}{\sigma_y^2}\|y-Ax\|^2$. It ensures that the objective function of Problem~\eqref{eq:equivariant_opt_pblm} is $\mathcal{C}^{\infty}$. Under the so-called manifold hypothesis, i.e. $p$ is supported on a compact, it is shown in~\cite{debortoli2023convergence} that  Assumption~\ref{ass:prior_score_approx}(c) is verified with $n_1 = 1$ and $\beta = -2$.

\begin{assumption}\label{ass:finite_moment}
    \textbf{(a)} The random variable $\fJ_G$ has a uniform finite moment, i.e. $\exists \epsilon > 0, M_{2+\epsilon} \ge 0$ such that $\forall x \in \R^d, \eE_{G \sim \pi}(\vvvert \fJ_G(x) \vvvert^{2+\epsilon}) \le M_{2+\epsilon} < +\infty$, with $\vvvert \cdot \vvvert$ the operator norm defined for $A \in \R^{d\times d}$ by $\vvvert A \vvvert = \sup_{\|x\|=1} \frac{\|Ax\|}{\|x\|}$.

\noindent
    \textbf{(b)} The transformation has bounded moments on any compact, i.e. $\forall\K\subset \R^d$ compact, $\forall m \in \N$, $\exists C_{\K, m}< + \infty$ such that $\forall x \in \K, \eE_{G \sim \pi}(\|  G(x) \|^m) \le C_{\K, m}$.
\end{assumption}
With Assumption~\ref{ass:finite_moment}, the behavior of the algorithm is controlled on each compact set. 
This assumption is verified for all examples presented in Section~\ref{sec:equivariant_ex}.

We now define $\mathbf{S}_{\sigma} = \{x \in \R^d | \nabla \fF_{\sigma}^{\pi}(x) = 0\}$, the set of critical points of $\fF_{\sigma}^{\pi}$ in $\R^d$, and $\Lambda_{\K}$, the set of random seeds for which the iterates of the algorithm are bounded in the compact $K$, by
\begin{equation*}
    \Lambda_{\K} = \bigcap_{k \in \N}{\{x_k \in \K\}}.
\end{equation*}
We finally denote the distance of a point to a set by $d(x, \mathbf{S}) = \inf_{y \in \mathbf{S}}{\|x - y \|}$, with $x \in \R^d$ and $\mathbf{S} \subset \R^d$.
The restriction to realizations of the algorithm bounded in $\Lambda_{K}$ will be referred to as the \textit{boundedness assumption}.

\begin{proposition}\label{prop:convergence_unbiased}
    Let $(x_k)_{k \in \N}$ be the iterates generated by Algorithm~\ref{alg:ERED} with the exact MMSE Denoiser $D^*_\sigma$. Then, under Assumptions~\ref{ass:step_size_decreas}-\ref{ass:finite_moment}, we have almost surely on $\Lambda_{\K}$
    \begin{align}
        &\lim_{k \to + \infty}{d(x_k, \mathbf{S_{\sigma}})} = 0,\label{eq1a}\\
        &\lim_{k \to + \infty}{\|\nabla \fF_{\sigma}^{\pi}(x_k)\|} = 0,\label{eq1b}
    \end{align}
and $(\fF_{\sigma}^{\pi}(x_k))_{k \in \N}$ converges to a value of $\fF_{\sigma}^{\pi}(\mathbf{S_{\sigma}})$.
\end{proposition}

Proposition~\ref{prop:convergence_unbiased} is proved in Appendix~\ref{sec:proof_convergence_unbiased}.

\subsection{ERED biased convergence analysis}
In this section, we analyse the convergence of the ERED algorithm (Algorithm~\ref{alg:ERED}) run with a realistic denoiser $D_{\sigma} \neq D_{\sigma}^\ast$. In this case, ERED is a biased stochastic gradient descent for solving Problem~\eqref{eq:equivariant_opt_pblm}. At each iteration, the  algorithm writes

\begin{align}\nonumber
    x_{k+1} = &\, x_k - \delta_k \nabla f(x_k)  \\&\hspace{-2pt}- \frac{\delta_k \lambda}{\sigma^2} \fJ_{G}^T(x_k) \left( G(x_k)\hspace{-2pt}-\hspace{-2pt} D_{\sigma}(G(x_k))\right),\label{eq:algo_biased}
\end{align}
with $G \sim \pi$. Defining the gradient estimator 
\begin{align*}
\xi_{k} =&\; \nabla f(x_k) + \frac{\lambda}{\sigma^2} \fJ_{G}^T(x_k) \left( G(x_k) - D_{\sigma}(G(x_k))\right) \\&- \nabla \fF_{\sigma}^{\pi}(x_k),
\end{align*}the algorithm~\eqref{eq:algo_biased} can be reformulated as 
\begin{align}
    x_{k+1} = x_k - \delta_k \left( \nabla \fF_{\sigma}^{\pi}(x_k) + \xi_k \right).
\end{align}

\begin{assumption}\label{ass:denoiser_sub_polynomial}
    The realistic denoiser $D_{\sigma}$ is sub-polynomial, i.e. $\exists C > 0$ and $n_2 \in N$ such that $\forall x \in \R^d$, $\|D_{\sigma}(x)\| \le C (1 + \|x\|^{n_2})$.
\end{assumption}

Assumption~\ref{ass:denoiser_sub_polynomial} is similar to Assumption~\ref{ass:prior_score_approx}(c) but it applies on the denoiser instead of the score of the underlying prior distribution. Notice that in this case the constant $C$ might depend on the noise level $\sigma$. 
As an example, a bounded denoiser~\cite{chan2016plug} verifies Assumption~\ref{ass:denoiser_sub_polynomial} with $n_2 = 1$.

\begin{assumption}
    \label{ass:g_bounded}
    For every compact $\K$, there exists $C_{\K}$, such that 
    $\forall x \in \K$, $\forall g \in \G, \|g(x)\| \le C_{\K}$. 
\end{assumption}
Assumption~\ref{ass:g_bounded} is verified as soon as $(g, x) \in \G \times \R^d \to g(x) \in \R^d$ is continuous and $\G$ compact. It is verified in particular for $\G$ being a finite set of isometries.

\begin{proposition}\label{prop:convergence_biaised}
    Let  $(x_k)_{k \in \N}$ be the sequence provided by ERED (Algorithm~\ref{alg:ERED}) with an inexact denoiser $D_{\sigma}$. Then,  under Assumptions~\ref{ass:step_size_decreas}-\ref{ass:denoiser_sub_polynomial},  there exists $M_{\K}$ such that, almost surely on $\Lambda_{\K}$:
    \begin{align}
        \limsup_{k \to \infty} \|\nabla \fF_{\sigma}^{\pi}(x_k)\| &\le M_{\K} \eta^{\frac{1}{2}} \label{eq1}\\
        \limsup_{k \to \infty} \fF_{\sigma}^{\pi}(x_k) - \liminf_{k \to \infty} \fF_{\sigma}^{\pi}(x_k) &\le M_{\K} \eta,\label{eq2}
    \end{align}
    with the asymptotic bias $\eta = \limsup_{k \to \infty} \|\eE(\xi_k)\|$.

    Moreover, under Assumption~\ref{ass:g_bounded}, we have
    \begin{align}\label{ineq}
        \eta \le \frac{\lambda}{\sigma^2} \sup_{x \in \K}\eE \left( \vvvert \fJ_G(x)\vvvert\right) \|D_{\sigma} - D_{\sigma}^\ast\|_{\infty, \L},
    \end{align}
    with $\L = \mathcal{B}(0, C_{\K})$, where $C_{\K}$ is introduced in Assumption~\ref{ass:g_bounded}.
\end{proposition}
Proposition~\ref{prop:convergence_biaised} is proved in Appendix~\ref{sec:proof_convergence_biaised}.

\subsection{ERED critical points analysis - Geometrical invariant case}
A critical point analysis for the noising-denoising case, presented in equation~\eqref{eq:noising_denoising_denoiser}, is provided in~\cite{renaud2024plug}. 
However, to the best of our knowledge, no such analysis has been provided so far even in the case of geometrical invariance, i.e. $p \circ g = p$ for $g \in \G$. Here, we fill this gap by studying critical points under the relaxed $\pi$-equivariance property.
First, we study the approximation $- \nabla \log p \approx s_{\sigma}^{\pi}$ when $\sigma \to 0$. Next, we deduce that critical points of Problem~\eqref{eq:equivariant_opt_pblm} converge to critical points of Problem~\eqref{eq:ideal_opt_pblm} when $\sigma \to 0$.

\begin{assumption}
     \textbf{(a)} \label{ass:p_regularity} The prior distribution $p \in \mathrm{C}^1(\R^d, ]0, +\infty[)$ with $\|p\|_{\infty} + \|\nabla p\|_{\infty} < + \infty$.
\noindent
    \textbf{(b)} \label{ass:first_moment_finite}
    $\fJ_G$ has finite first moment, i.e. $\sup_{x \in \R^d}\eE_{G \sim \pi}(\vvvert \fJ_G(x)\vvvert) < +\infty$.
\end{assumption}
Assumption~\ref{ass:p_regularity}(a) is needed to ensure that $\nabla \log p$ is well defined.
Assumption~\ref{ass:first_moment_finite}(b) is verified in particular for a finite set of transformations, for a set of linear isometries or for the noising-denoising regularization.

\begin{proposition}\label{prop:uniform_cvg}
    Under Assumptions~\ref{ass:g_bounded}-\ref{ass:first_moment_finite}, for every compact $\K \subset \R^d$, if the prior $p$ is $\pi$-equivariant, we have, when $\sigma \to 0$,
    \begin{equation}
        \|s - s_{\sigma}^{\pi}\|_{\infty, \K} \to 0.
    \end{equation}
\end{proposition}

Proposition~\ref{prop:uniform_cvg} is proved in Appendix~\ref{sec:proof_uniform_cvg}.
From Proposition~\ref{prop:uniform_cvg}, we now deduce  a  critical point convergence when $\sigma \to 0$ in the sense of Kuratowski~\cite{chambolle20231d}.

\begin{assumption}\label{ass:data_fidelity_diff}
    The data-fidelity term in~\eqref{eq:ideal_opt_pblm} is continuously differentiable, i.e. $f \in \mathcal{C}^1(\R^d, \R)$.
\end{assumption}
Assumption~\ref{ass:data_fidelity_diff} is needed to define the critical points of Problem~\eqref{eq:ideal_opt_pblm}. It is verified for a large set of inverse problems, including linear inverse problems with Gaussian noise, phase retrieval or despeckling.

We denote by $\mathbf{S}^\ast$ the set of critical point of $\fF$, 
$\mathbf{S}_{\sigma}$ the set of critical points of $\fF_{\sigma}^{\pi}$,  and $\mathbf{S}$ the limit point of $\mathbf{S}_{\sigma}$ when $\sigma \to 0$, more precisely
\begin{align}
    \mathbf{S} = \{ &x \in \R^d | \exists \sigma_n>0 \text{ decreasing to } 0,\nonumber\\
    &x_n \in \mathbf{S}_{\sigma_n} \text{ such that } x_n \xrightarrow[n \to \infty]{} x\}.
\end{align}

\begin{proposition}\label{prop:critical_points_cvg}
    Under Assumptions~\ref{ass:g_bounded}-\ref{ass:data_fidelity_diff}, if the prior $p$ is $\pi$-equivariant, we have
    \begin{align}
        \mathbf{S} \subset \mathbf{S}^\ast.
    \end{align}
\end{proposition}

Proposition~\ref{prop:critical_points_cvg} is proved in Appendix~\ref{sec:proof_critical_points_cvg}.
\vspace*{-0.25cm}

\section{Equivariant Plug-and-Play (EPnP)}\label{sec:epnp}

In the previous section, we have studied ERED (Algorithm~\ref{alg:ERED}) and proved its convergence. In this algorithm, the pretrained denoiser is used to approximate a gradient-step on the regularization. In another line of research, named Plug-and-Play~\cite{venkatakrishnan2013plug} or PnP, a Gaussian denoiser approximates a proximal-step on the regularization. In section~\ref{ssec:epnp}, we first introduce the PnP counterpart of ERED, that we call Equivariant Plug-and-Play (EPnP).
EPnP is then reformulated as a Perturbed proximal gradient descent algorithm in order to analyse its convergence in section~\ref{ssec: conv_epnp}.

\subsection{EPnP: a Perturbed Proximal Gradient Descent algorithm}\label{ssec:epnp}

We first propose to inject the equivariant denoiser $\tilde D_{\sigma}$ defined in equation~\eqref{eq:equiv_denoiser} into the Plug-and-Play algorithm, leading to Equivariant Plug-and-Play (EPnP, Algorithm~\ref{alg:EPnP}). This algorithm was previously proposed in~\cite{terris2024equivariant} in the particular case of $\G$ being a finite set of isometries.

\begin{algorithm}
\caption{EPnP}\label{alg:EPnP}
\begin{algorithmic}[1]
\STATE \textbf{Parameters:} $x_0 \in \R^d$, $\sigma > 0$, $\lambda > 0$, $N \in \N$
\STATE \textbf{Input:} degraded image $y$
\STATE \textbf{Output:} restored image $x_{N}$
\FOR{$k = 0, 1, \dots, N-1$}
    \STATE Sample $G \sim \pi$
    \STATE $y_k = x_k - \frac{1}{\lambda} \nabla f(x_k)$
    \STATE $x_{k+1} = \fJ_{G}^T(y_k) D_{\sigma}\left(G\left(y_k \right)\right)$
\ENDFOR
\end{algorithmic}
\end{algorithm}

We now study a setting where EPnP
can be reformulated as a Perturbed Proximal Gradient Descent (PGD). To do so, we first need the denoiser to induce a proximal step on a regularization.

\begin{assumption}\label{ass:denoiser_structure}
The denoiser is a gradient-step, i.e. there exists $h_{\sigma}: \R^d \to \R$ such that $D_{\sigma} = I_d - \nabla h_{\sigma}$. Moreover, $h_{\sigma}$ is $\mathcal{C}^2$ and $\nabla h_{\sigma}$ is $L_h$-Lipschitz with $L_h < 1$.
\end{assumption}
We then rely on the denoiser architecture  $D_{\sigma}$ verifying Assumption~\ref{ass:denoiser_structure} and proposed in ~\cite{hurault2024convergent}.

\begin{proposition}{\cite[Proposition 5]{hurault2024convergent}}\label{prop:denoiser_is_a_prox}
Under Assumption~\ref{ass:denoiser_structure}, there exists a  potential $g_{\sigma}$ such that
\begin{align}
    D_{\sigma} = \mathsf{Prox}_{g_{\sigma}},
\end{align}
where $g_{\sigma}$ is $\frac{L_h}{L_h+1}$-weakly convex.
\end{proposition}

By definition, the iterates of EPnP $x_k$ belong to $\text{Im}(D_{\sigma})$ for $k \ge 1$. 
Note that $\text{Im}(D_{\sigma})$ is not necessary convex for a deep neural network $D_{\sigma}$ due to the non-linearity of activations. It is nevertheless shown in~\cite[Lemma 1(iv)]{renaud2025stability}  that $\text{Im}(D_{\sigma})$ is open and almost convex in the sense that $\text{Leb}\left( \text{Conv}(\text{Im}(D_{\sigma})) \setminus \text{Im}(D_{\sigma})\right) = 0$, where $\text{Leb}$ is the Lebesgue measure and $\text{Conv}(\text{Im}(D_{\sigma}))$ the convex hull of $\text{Im}(D_{\sigma})$.

Under Assumption~\ref{ass:denoiser_structure}, EPnP can be written as
\begin{align}
    x_{k+1} &=  \fJ_{G}^T(y_k) \left(I_d - \nabla h_{\sigma} \right) \left(G \left( y_k \right) \right) \nonumber \\
    &= \fJ_{G}^T(y_k) G(y_k) - \fJ_{G}^T(y_k) \nabla h_{\sigma}(G(y_k)) \nonumber \\
    &= \fJ_{G}^T(y_k) G(y_k) - \nabla (h_{\sigma} \circ G)(y_k)\label{eq:first_part_reformulation_epnp},
\end{align}
with $G \sim \pi$ a random transformation of $\R^d$ and $y_k = x_k - \frac{1}{\lambda}\nabla f(x_k)$.

\begin{assumption}\label{ass:sructure_on_transf}
All $g \in \G$ are affine transformations with isometric linear parts, i.e. $\forall g \in \G$, there exist $a : \R^d \to \R^d$ a linear isometry and $c \in \R^d$, such that $\forall x \in \R^d$, $g(x) = a x + c$. Moreover, the distribution $\pi$ is unbiased and has uniformly bounded variance, i.e. $\eE_{\pi}\left(J_G^T(x)G(x) \right) = x$ and there exists $\mu \ge 0$ such that $\forall x \in \R^d$, $\eE_{\pi}\left(\|G(x) - \eE_\pi(G(x))\|^2\right) \le \mu^2$.
\end{assumption}
Assumption~\ref{ass:sructure_on_transf} is verified for all sets of transformations described in Section~\ref{sec:equivariant_ex}.

Under Assumptions~\ref{ass:sructure_on_transf}, for all $g \in \G$, we have that $\fJ_g = a$ and $\fJ_g^T = a^{-1}$ and we can reformulate equation~\eqref{eq:first_part_reformulation_epnp} as
\begin{align*}
    x_{k+1} = y_k - A^{-1} \nabla h_{\sigma}(A y_k + C),
\end{align*}
where we decomposed the random variable $G(x) = Ax + C$, with $A$ a random linear isometry and $C$ a random constant.
The function $x \in \R^d \mapsto a^{-1}\nabla h_{\sigma}(a x +c)$ is $L_h$-Lipschitz: for $x, y \in \R^d$
\begin{align*}
    &\|a^{-1}\nabla h_{\sigma}(a x +c) -  a^{-1}\nabla h_{\sigma}(a y +c)\| \\\le& \normiii{a^{-1}} \|\nabla h_{\sigma}(a x +c) - \nabla h_{\sigma}(a y +c)\| \\
    \le& L_h \|a x  - a y\| = L_h \|x  - y\|,
\end{align*}
because $a$ is an isometry and thus $\normiii{a^{-1}} = 1$. Therefore, we can apply Proposition~\ref{prop:denoiser_is_a_prox} and there exists $\tilde g_{\sigma}$ such that 
\begin{align}\nonumber
    \tilde D_{\sigma} 
    &= \eE_{G \sim \pi}\left(D_{\sigma} \circ G\right)\\
    &= I_d - \eE_{G \sim \pi}\left(\nabla (h_{\sigma} \circ G ) \right) = \mathsf{Prox}_{\tilde g_{\sigma}},
\end{align}
with $\tilde D_{\sigma}$ the equivariant denoiser defined in equation~\eqref{eq:equiv_denoiser}. 
Note that the regularization $\tilde g_{\sigma}$  underlying $\tilde D_{\sigma}$ is different from the regularization $g_{\sigma}$ underlying $D_{\sigma}$.

However, at each iteration, we do not compute the equivariant denoiser but only a stochastic estimation of it, with one random transformation $G$. We define this estimator for $x \in \R^d$ by
\begin{align}\nonumber
    \widehat{\mathsf{Prox}}_{\tilde g_{\sigma}}(x) &= x - \nabla \left(h_{\sigma} \circ G\right) (x)\\& = x - A^{-1} \nabla h_{\sigma}(G(x)).
\end{align}
This estimator is unbiased by definition: $\eE_{G \sim \pi}\left(\widehat{\mathsf{Prox}}_{\tilde g_{\sigma}}(x)\right) = \Prox{\tilde g_{\sigma}}{x}$.

Therefore, under Assumptions~\ref{ass:denoiser_structure}-\ref{ass:sructure_on_transf}, EPnP can be reformulated as
\begin{align*}
    x_{k+1} = \widehat{\mathsf{Prox}}_{\tilde g_{\sigma}}\left(x_k - \frac{1}{\lambda} \nabla f(x_k)\right),
\end{align*}
or, equivalently,
\begin{align}\label{eq:perturbe_version_snopnp}
    x_{k+1} &=  \mathsf{Prox}_{\tilde g_{\sigma}}\left( x_k - \frac{1}{\lambda} \nabla f(x_{k}) \right) + \zeta_{k+1},
\end{align}
with $\zeta_{k+1} = \widehat{\mathsf{Prox}}_{\tilde g_{\sigma}}\left( x_k - \delta \nabla f(x_{k}) \right) - \mathsf{Prox}_{\tilde g_{\sigma}}\left( x_k - \delta \nabla f(x_{k}) \right)$. 
The previous algorithm is a Perturbed PGD (PPGD) in the sense that we only access to the proximal operator up to a random perturbation $\zeta_k$.

Notice that the perturbation $\zeta_{k+1}$ is centered, as  $\eE(\zeta_{k+1}|x_k) = 0$ and it has a finite variance
\begin{align}\nonumber
    &\eE(\|\zeta_{k+1}\|^2 | x_k)\\
    \nonumber =&\, \eE\left(\left\|y_k - A^{-1}\nabla h_{\sigma}(G(y_k)) - y_k\right.\right. \\
    \nonumber &\left.\left.+ \eE\left(A^{-1}\nabla h_{\sigma}(G(y_k))\right)\right\|^2\right) \\
    \le&\, \eE(\|\nabla h_{\sigma}(G(y_k)) - \eE\left(\nabla h_{\sigma}(G(y_k))\right)\|^2) \nonumber\\
   \le&\, L_h^2 \eE_{G_1, G_1 \sim \pi}(\|G_1(y_k) - G_2(y_k)\|^2) \nonumber\\
    \le&\,2 L_h^2 \mu^2.\label{eq:variance_zeta}
\end{align}

Hence equation~\eqref{eq:perturbe_version_snopnp} shows that EPnP can be seen as a PPGD algorithm to solve the following optimization problem
\begin{equation}\label{eq:opt_problem_perturbed}
    \argmin_{x \in \R^d} \tilde \fF(x) := f(x) + \lambda \tilde g_{\sigma}(x).
\end{equation}

\subsection{Convergence analysis of EPnP}\label{ssec: conv_epnp}

In the previous section, EPnP is formulated as a Perturbed PGD. To our knowledge such a perturbed evaluation of a proximal operator has not been studied in the literature. There only exist analyses of inexact proximal operators~\cite{salim2019snake,patrascu2021stochastic}
Therefore, in this section we provide the first proof of convergence of a Perturbed PGD algorithm, namely EPnP.

\begin{assumption}\label{ass:regularities_assumptions}
    \noindent
    \textbf{(a)} $f$ is differentiable and $\nabla f$ is $L_f$ Lipschitz.

    \noindent
    \textbf{(b)} $\tilde \fF$ admit a lower bound $\fF^\ast \in \R$, i.e. $\forall x \in \R^d, \tilde \fF(x) \ge \fF^{\ast}$.
\end{assumption}

Assumption~\ref{ass:regularities_assumptions}(a) is verified for a large class of inverse problems, including linear degradation with additive Gaussian noise. Assumption~\ref{ass:regularities_assumptions}(b) is necessary for problem~\eqref{eq:opt_problem_snopnp} to be well defined and is verified in our practical applications.

\begin{lemma}\label{lemma:residual_control_perturbed_pgd}
Under Assumptions~\ref{ass:denoiser_structure}-\ref{ass:sructure_on_transf}-\ref{ass:regularities_assumptions}, for $\lambda \ge 3 L_f$, there exist $C_1, C_2 \in \R_+$ such that there iterates of EPnP (Algorithm~\ref{alg:EPnP}) verifies
\begin{align*}
    \frac{1}{N} \sum_{k=0}^{N-1} \eE(\|x_{k+1} - x_k\|^2) &\le C_1\frac{\tilde \fF(x_0) - \fF^\ast}{N} + C_2 \mu^2.
\end{align*}
\end{lemma}
Lemma~\ref{lemma:residual_control_perturbed_pgd} shows that we can control the residuals of the iterates of EPnP to solve problem~\eqref{eq:opt_problem_perturbed}. Its proof is provided in Appendix~\ref{sec:proox_residual_control_perturbed_pgd}

Now, we have all the tools to prove the main convergence result for EPnP.

\begin{proposition}\label{prop:cvg_critical_points_prox_perturbed}
 Under Assumption~\ref{ass:denoiser_structure} and~\ref{ass:regularities_assumptions}, for $\lambda \ge 3 L_f$, there exist $D_1, D_2 >0$ independent of $\sigma$ such that
 \begin{align*}
     \frac{1}{N}  \hspace{-1.5pt}\sum_{k=0}^{N-1} \eE(\|\nabla \tilde \fF(x_k - \zeta_k)\|^2) \le D_1 \frac{\tilde \fF(x_0)  \hspace{-1.5pt}- \hspace{-1.5pt}\fF^\ast}{N} + D_2 \mu^2 \hspace{-1.5pt}.
 \end{align*}
\end{proposition}

Proposition~\ref{prop:cvg_critical_points_prox_perturbed} shows that, up to a random noise $\zeta_k$, the iterates get close to critical points of problem~\eqref{eq:opt_problem_perturbed}. Moreover, the precision of the algorithm is controlled by the random transformation fluctuations with $\mu > 0$ introduced in Assumption~\ref{ass:sructure_on_transf}.

\begin{proof}
By the optimal condition of the proximal operator, we have from~\eqref{eq:perturbe_version_snopnp}
\begin{align*}
    x_{k+1} - x_k = \zeta_{k+1} - \frac{1}{\lambda} \nabla f(x_k) - \nabla \tilde g_{\sigma}(x_{k+1} - \zeta_{k+1}).
\end{align*}

So, recalling that $\tilde \fF=f+\tilde g_\sigma$, we get
\begin{align*}
    &\|\nabla \tilde \fF(x_{k+1} - \zeta_{k+1})\|^2 = \lambda^2 \|\frac{1}{\lambda}\nabla \tilde \fF(x_{k+1} - \zeta_{k+1})\|^2 \\
    \le&\, 2\lambda^2 \|\frac{1}{\lambda}\nabla \tilde \fF(x_{k+1} - \zeta_{k+1}) - (x_{k+1} - x_k)\|^2 \\&+ 2 \lambda^2 \|x_{k+1} - x_k\|^2 \\
    \le&\,  2\lambda^2 \|\frac{1}{\lambda} \left(\nabla f(x_{k+1}- \zeta_{k+1}) - \nabla f(x_k) \right) + \zeta_{k+1} \|^2 \\&+ 2 \lambda^2 \|x_{k+1} - x_k\|^2 \\
    \le&\,  4 \|\nabla f(x_{k+1}- \zeta_{k+1}) - \nabla f(x_k) \|^2 + 4 \lambda^2 \|\zeta_{k+1}\|^2 \\&+ 2 \lambda^2 \|x_{k+1} - x_k\|^2 \\
    \le&\,  8 L_f \|x_{k+1} - x_k \|^2 + 8 L_f \|\zeta_{k+1} \|^2 + 4 \lambda^2 \|\zeta_{k+1}\|^2 \\&+ 2 \lambda^2 \|x_{k+1} - x_k\|^2 \\
   \le&\,  (8 L_f + 2 \lambda^2) \|x_{k+1} - x_k\|^2 + (8 L_f + 4 \lambda^2) \|\zeta_{k+1}\|^2.
\end{align*}

By taking the expectation in the previous inequality, averaging between 0 and $N-1$,  using relation~\eqref{eq:variance_zeta} and Lemma~\ref{lemma:residual_control_perturbed_pgd}, we get
\begin{align*}
    &\frac{1}{N} \sum_{k=0}^{N-1} \eE(\|\nabla \tilde \fF(x_k - \zeta_k)\|^2) \\\le&\, (4 L_f + 2 \lambda^2) \frac{1}{N} \sum_{k=0}^{N-1}\eE(\|x_{k+1} - x_k\|^2) \\
    &+ 2(8 L_f + 4 \lambda^2) L_h^2 \mu^2 \\
    \le&\, (4 L_f + 2 \lambda^2) C_1 \frac{\tilde \fF(x_0) - \fF^\ast}{N} \\&+ \left((4 L_f + 2 \lambda^2)C_2 + 2(8 L_f + 4 \lambda^2) L_h^2\right) \mu^2 \\
    \le&\, D_1 \frac{\tilde \fF(x_0) - \fF^\ast}{N} + D_2 \mu^2,
\end{align*}
with the constants $D_1 = \frac{2(4 L_f + 2 \lambda^2)}{(\lambda (2 - \rho) - 3L_f)}$ and $D_2 = (4 L_f + 2 \lambda^2)\frac{2(11 L_f + 8 \rho \lambda + 2 \lambda)}{\lambda (2 - \rho) - 3L_f} L_h^2 + 2(8 L_f + 4 \lambda^2) L_h^2$.
\end{proof}

\section{Stochastic denoising Plug-and-Play (SnoPnP)}\label{sec:snopnp}

A particularly interesting case of EPnP (Algorithm~\ref{alg:EPnP}) is the case where the transformation used is a random translation with a Gaussian-distributed shift (see \textit{Noising-denoising} paragraph in Section~\ref{sec:equivariant_ex}). In this case, the Equivariant denoiser consists in noising the iterate and denoising it.
We name this algorithm Stochastic denoising Plug-and-Play (SnoPnP, Algorithm~\ref{alg:SnoPnP}) which can be seen as a Plug-and-Play counterpart of SNORE~\cite{renaud2024plugandplayimagerestorationstochastic}, with an approximate proximal step, instead of a gradient step, realized on the regularization. 
\begin{algorithm}
\caption{SnoPnP}\label{alg:SnoPnP}
\begin{algorithmic}[1]
\STATE \textbf{Parameters:} $x_0 \in \R^d$, $\sigma > 0$, $\lambda > 0$, $N \in \N$
\STATE \textbf{Input:} degraded image $y$
\STATE \textbf{Output:} restored image $x_{N}$
\FOR{$k = 0, 1, \dots, N-1$}
    \STATE Sample $z_{k+1} \sim \mathcal{N}(0, I_d)$
    \STATE $x_{k+1} = D_{\sigma}\left(x_k - \frac{1}{\lambda} \nabla f(x_k) + \sigma z_{k+1}\right)$
\ENDFOR
\end{algorithmic}
\end{algorithm}

For SnoPnP, which is a particular case of EPnP, we can refine the convergence analysis presented in Section~\ref{sec:epnp}.
Under Assumption~\ref{ass:denoiser_structure}, SnoPnP can be written as
\begin{align}\label{eq:snopnp_with_denoiser_structure}
    x_{k+1} &=  \mathsf{Prox}_{g_{\sigma}}\left( x_k - \frac{1}{\lambda} \nabla f(x_k) + \sigma z_{k+1} \right).
\end{align}

By denoting $\tilde \nabla f(x_k) = \nabla f(x_k) - \sigma \lambda z_{k+1}$,
which is an unbiased stochastic approximation of $\nabla f(x_k)$,
SnoPnP appears to be a Stochastic Proximal Gradient Descent (SPGD) algorithm,
\begin{align}\label{eq:snopnp_iterates}
    x_{k+1} &=  \mathsf{Prox}_{g_{\sigma}}\left( x_k - \frac{1}{\lambda} \tilde \nabla f(x_{k}) \right),
\end{align}
to solve the optimization problem defined by
\begin{align}\label{eq:opt_problem_snopnp}
    \argmin_{x \in \R^d} F(x) := f(x)+ \lambda g_{\sigma}(x) 
\end{align}

Stochastic proximal gradient descent algorithms have been extensively analyzed with a convex regularization, i.e. $g_{\sigma}$  convex~\cite{xiao2014proximal,nitanda2014stochastic,j2016proximal,ghadimi2016mini,atchade2017perturbed,allen2018katyusha,ding2023nonconvex,xu2023momentum}, and more recently with weakly convex regularizations~\cite{li2022unified,renaud2024convergenceanalysisproximalstochastic}.
Thanks to Proposition~\ref{prop:denoiser_is_a_prox}, we work here in a weakly convex setting.

\begin{lemma}[\cite{renaud2024convergenceanalysisproximalstochastic}]\label{lemma:residual_control_snopnp}
Under Assumptions~\ref{ass:denoiser_structure} and~\ref{ass:regularities_assumptions}, for $\lambda \ge L_f$, there exist $A_1, A_2 \in \R_+$ such that
\begin{align}
    \frac{1}{N} \hspace{-1.5pt}\sum_{k=0}^{N-1} \eE_k\hspace{-1.5pt}\left(\|x_{k+1} - x_k\|^2\right) \hspace{-1.5pt}\le \frac{A_1(F(x_0)\hspace{-1.pt} -\hspace{-1.5pt} F^\ast)}{N}\hspace{-1.5pt} +\hspace{-1.5pt}  A_2 \sigma^2.
\end{align}
\end{lemma}

Lemma~\ref{lemma:residual_control_snopnp} provides a control on the residual of the iterates but does not give any convergence guarantees. For completeness, the proof is recalled in Appendix~\ref{sec:proof_residuals_snopnp}. With more regularity on the regularization $g_{\sigma}$, we can derive a convergence result on $\nabla F(x_k)$.

\begin{assumption}\label{ass:g_l_smooth}
     $g_{\sigma}$ is $L_g$-smooth on $\mathsf{Im}(D_{\sigma})$, i.e. $\nabla g_{\sigma}$ is $L_g$-Lipschitz on $\mathsf{Im}(D_{\sigma})$.
\end{assumption}

Assumption~\ref{ass:g_l_smooth} allows to control the behavior of $g_{\sigma}$ on the domain of interest $\mathsf{Im}(D_{\sigma})$. Let us recall here that $\mathsf{Im}(D_{\sigma})$ is an open set because of the image of the proximal operator of a weakly convex function, see~\cite[Proposition 4]{renaud2025moreau}. It is verified by the denoiser proposed in~\cite{hurault2024convergent}.

\begin{proposition}\label{prop:snopnp_convergence_critical_point}
Under Assumption~\ref{ass:denoiser_structure}, \ref{ass:regularities_assumptions} and~\ref{ass:g_l_smooth}, for $\lambda \ge L_f$, there exist $B_1, B_2 \in \R_+$ independent of $\sigma$, such that
\begin{align}\label{eq:prop_snopnp_cvg_critical_point}
    \frac{1}{N} \sum_{k=0}^{N-1}\eE\left(\|\nabla F(x_k)\|^2\right)
    \le \frac{B_1 (F(x_0) - F^\ast)}{N} + B_2 \sigma^2.
\end{align}
\end{proposition}

Proposition~\ref{prop:snopnp_convergence_critical_point} is shown in Appendix~\ref{sec:appendix_proof_proposition_cvg_critical_points_spgd}.
We interpret this result as the iterates $x_k$ convergence to a level of $\nabla F$ parameterized by $\sigma$, the quantity of noise being added at each iteration. Note that Proposition~\ref{prop:snopnp_convergence_critical_point} refines Proposition~\ref{prop:cvg_critical_points_prox_perturbed} in two ways. First, we directly target the values of $\nabla F$ without additional noise. Secondly, the constraint on the regularization parameter $\lambda$ is refined by a factor $3$.

\begin{remark}
    Proposition~\ref{prop:snopnp_convergence_critical_point} proves that SnoPnP (Algorithm~\ref{alg:SnoPnP}) targets critical points of Problem~\ref{eq:opt_problem_snopnp}. Moreover, it is shown in~\cite{renaud2025stability}  that the algorithm samples the posterior distribution when the standard-deviation of the injected noise is $\sqrt{2} \sigma$ instead of $\sigma$ in SnoPnP. The series of works~\cite{Laumont_2022,laumont2023maximum} exhibit the same result when considering an explicit gradient step on the regularization instead of an implicit one as in SnoPnP.
\end{remark}

\begin{assumption}\label{ass:F_subanalytic}
     \textbf{(a)} $F$ is subanalytic (see definition in Appendix~\ref{sec:def_subanalytic}) on $\R^d$.
     
     \textbf{(b)} $\nabla F$ is coercive, i.e. $\lim_{\|x\|\to+\infty}\|\nabla F(x)\| = +\infty$.
\end{assumption}

Assumption~\ref{ass:F_subanalytic}(a) is verified for a large class of functions including all standard functions, i.e. compositions of polynomials, exponential or logarithmic functions. 
Assumption~\ref{ass:F_subanalytic}(b) seems hard to ensure in practice as it relies on a subtle interaction between the data-fidelity and the regularization.

We recall that we denote by $\sS^\ast = \{x \in \R^d | \nabla F(x) = 0\}$ the set of critical points of $F$ and by $d(x, \sS^\ast) = \inf_{y \in \sS^\ast}\|x-y\|$ the distance of $x$ to the set of critical points of $F$.

\begin{corollary}\label{cor:almost_sre_cvg_subsequence}
Under Assumptions~\ref{ass:denoiser_structure}, \ref{ass:regularities_assumptions}, \ref{ass:g_l_smooth} and~\ref{ass:F_subanalytic}, 
for any $\lambda \ge L_f$ and $\beta > 0$, with probability larger than $1 - \beta$ there exist a subsequence $x_{\psi(k)}$, $B_3 > 0$ and $r > 0$ such that
\begin{align}
    d(x_{\psi(k)}, \sS^\ast) \le B_3 \left(\frac{\sigma}{\beta}\right)^r.
\end{align}
\end{corollary}
Corollary~\ref{cor:almost_sre_cvg_subsequence} is proved in Appendix~\ref{sec:proof_cor_almost_sure_cvg_subsequence}.
It shows that under additional geometrical properties on $F$ we can deduce from Proposition~\ref{prop:snopnp_convergence_critical_point} that a subsequence of the SnoPnP algorithm gets close to critical points of $F$. This formalizes the statement that Proposition~\ref{prop:snopnp_convergence_critical_point} proves a convergence to a critical point of $F$.

\begin{figure*}[t]
    \centering
    \includegraphics[width=\textwidth]{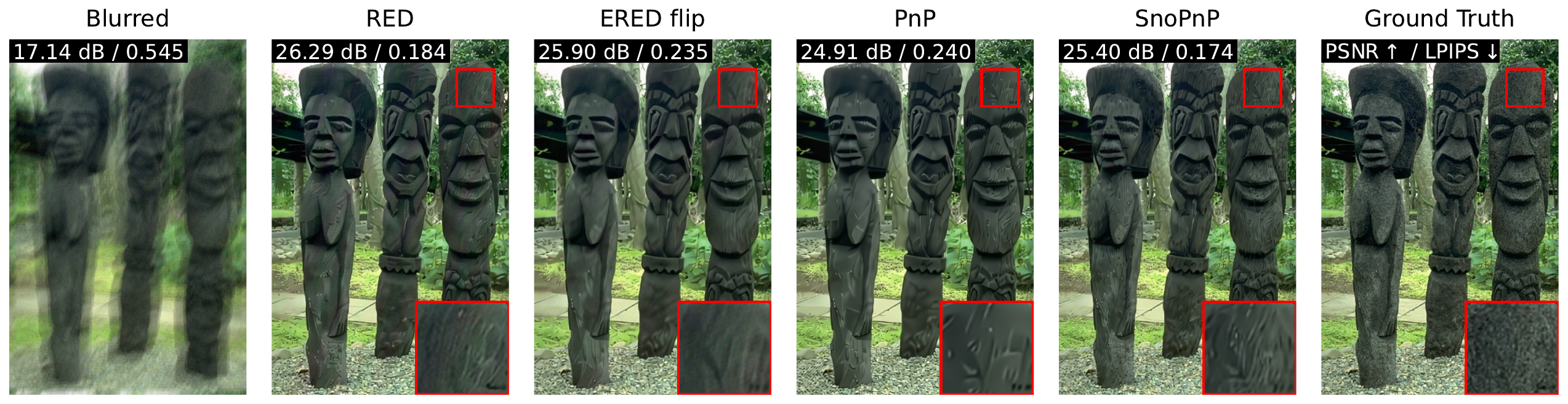}
    \caption{Deblurring (with a motion blur kernel and input noise level $\sigma_{y} = 5 / 255$) with various methods based on a GS-denoiser trained on natural images. Note that none of the methods does succeed to restore the grain of the texture.  
    \vspace*{-0.4cm}}
    \label{fig:deblurring_qualitative}
\end{figure*}

\section{Experiments}\label{sec:experimental}
In this section, we evaluate the practical gain of the three presented equivariant algorithms (Algorithm~\ref{alg:ERED}-\ref{alg:EPnP}-\ref{alg:SnoPnP}) for image restoration. Hyper parameters choices are provided in Appendix~\ref{sec:hyperparameter}.

\subsection{Deblurring}
First we evaluate various equivariance methods on image deblurring with $10$ blur kernels including fixed and motion kernels as proposed by~\cite{romano2017little,hurault2022gradient}. The denoiser used in these experiments is the Gradient-Step DRUNet denoiser (GS-DRUNet) proposed by~\cite{hurault2022gradient} with the provided weights, obtained with supervised training on natural color images. This denoiser reaches state-of-the-art denoising performance.

\begin{table}[ht]
\centering
\resizebox{0.9\linewidth}{!}{
\begin{tabular}{ c c c c c }
 Method & PSNR$\uparrow$ & SSIM$\uparrow$ & LPIPS$\downarrow$ & N$\downarrow$\\
\hline
RED~\cite{romano2017little} & 29.84 & \underline{0.84} & \textbf{0.13} & 400 \\
ERED rotation~\cite{terris2024equivariant} & 30.01 & \underline{0.84} & 0.16 & 400 \\
ERED translation & 29.95 & \underline{0.84} & 0.16 & 400 \\
ERED flip & 30.00 & \underline{0.84} & 0.16 & 400 \\
ERED subpixel rotation & 29.80 & \underline{0.84} & 0.16 & 400\\
ERED all transformations & 29.71 & 0.82 & \underline{0.15} & 400 \\
SNORE~\cite{renaud2024plug} & 29.55 & 0.83 & 0.19 & 1000 \\
Annealed SNORE~\cite{renaud2024plug} & 29.92 & \textbf{0.85} & 0.17 & 1500\\
PnP~\cite{venkatakrishnan2013plug} & 29.98 & 0.84 & 0.18 & 400\\
EPnP rotation~\cite{terris2024equivariant} & \textbf{30.22} & 0.84 & 0.18 & 400 \\
EPnP translation & 29.99 & 0.84 & 0.18 & 400 \\
EPnP flip & \underline{30.16} & 0.84 & 0.18 & 400\\
EPnP subpixel rotation & 30.03 & 0.84 & 0.18 & 400 \\
EPnP all transformations & 30.15 & 0.84 & 0.17 & 400 \\
SnoPnP & 29.67 & 0.83 & 0.16 & 100\\
 \hline
\end{tabular}
} 
\caption{Quantitative comparison of image deblurring methods on $10$ images from CBSD68 dataset with $10$ different blur kernels (fixed and motion kernel of blur) and a noise level $\sigma_{y}  = 5 / 255$.
Best and second best results are respectively displayed in bold and underlined.}
\label{table:quantitative_results}
\end{table}

In Table~\ref{table:quantitative_results}, we present the performance reached by Algorithm~\ref{alg:ERED} and Algorithm~\ref{alg:EPnP} with various sets of transformations including random rotation of angle $\theta \in \{0,1,2,3\} \times \pi / 2$ named \textit{rotation}, random subpixel rotation of angle $\theta \in [-\pi, \pi]$ named \textit{subpixel rotation} (implemented with raster rotation~\cite{paeth1990fast}), random flip of the two axes named \textit{flip}, random translation along both pixel axes named \textit{translation}, random Gaussian noising named \textit{SNORE}~\cite{renaud2024plug}, and a random transformation drawn among all the previous sets of transformations (including SNORE) named \textit{all transformations}. We also present performance of RED and Annealed SNORE, another version of SNORE~\cite{renaud2024plug} where the denoiser parameter $\sigma$ decreases through iterations. We notice that ERED and EPnP obtains better quantitative performance than RED and PnP respectively ($+0.2$dB) with the same computational cost. Moreover, the choice of the transformation impacts the restoration quality and our experiments suggest that flips and rotations are beneficial.

In Figure~\ref{fig:deblurring_qualitative}, we observe that restoration methods (in particular PnP) hallucinate structure in pure texture areas. Moreover, equivariance help to reduce these structures and generate a more realistic texture.

\subsection{Convergence of SnoPnP: GS-DRUNet vs Prox-DRUNet}

In Sections~\ref{sec:epnp}-\ref{sec:snopnp}, we prove the convergence of SnoPnP under the assumption that the denoiser is a proximal operator, see Assumption~\ref{ass:denoiser_structure}. However, this hypothesis is not verified by the Gradient Step denoiser~\cite{hurault2022gradient}. 
In~\cite{hurault2022proximal}, a penalization of the training loss is then  proposed to obtain a denoiser, named Prox-DRUNet, that verifies this hypothesis.
In Figure~\ref{fig:prox_vs_gs_drunet}, we compare the convergence of SnoPnP with both denoisers. It supports our theoretical findings and the tightness of assumption on the denoiser. SnoPnP with Prox-DRUNet indeed converges whereas GS-DRUNet, being a gradient-step operator and not a proximal one, makes SnoPnP diverge. 

\begin{figure}
    \centering
    \includegraphics[width=1.\linewidth]{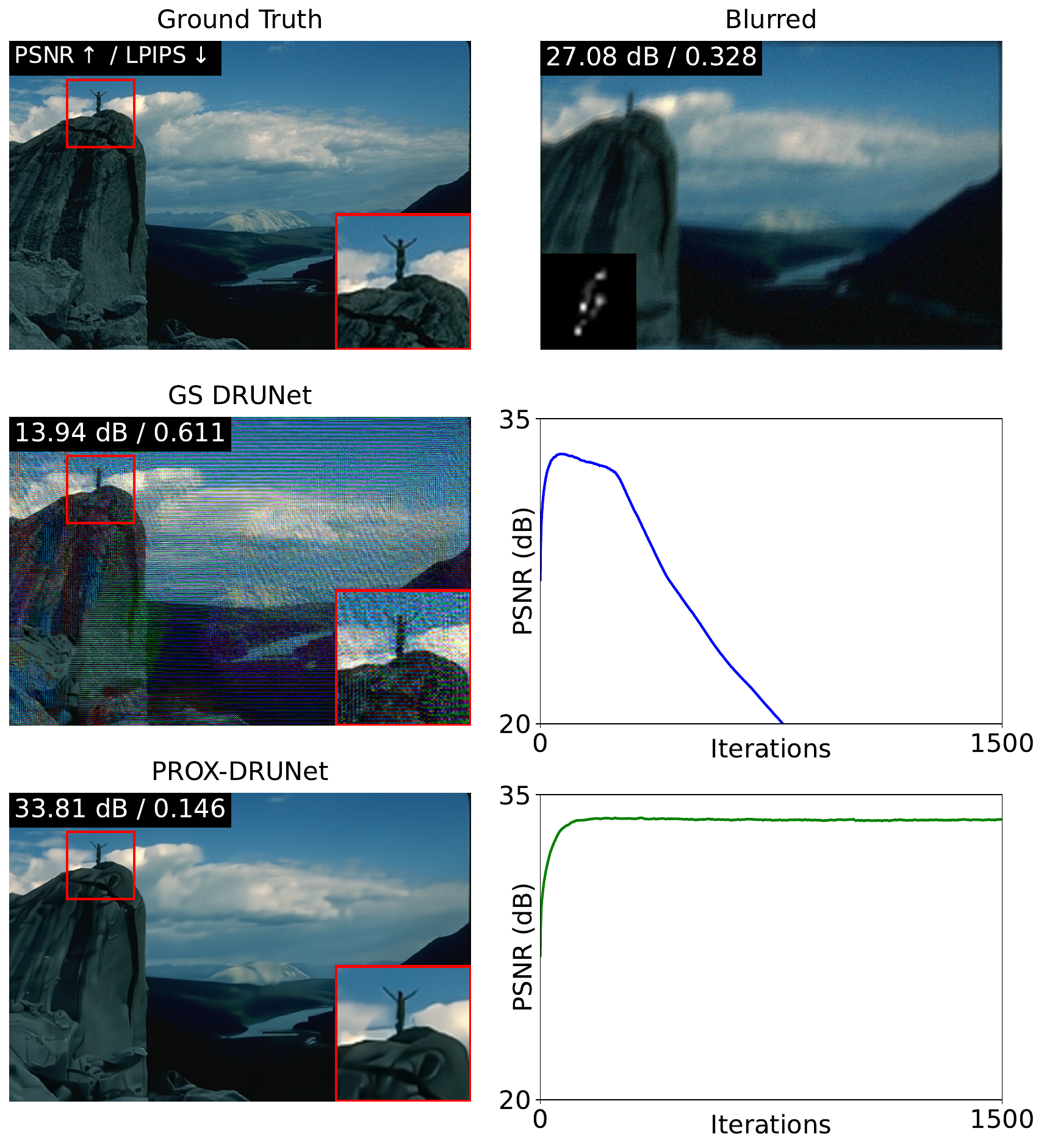}
    \caption{Deblurring (a motion blur kernel with input noise level $\sigma_y = 5/255$ with SnoPnP using a GS denoiser and a Prox Denoiser, both trained on natural images.}
    \label{fig:prox_vs_gs_drunet}
\end{figure}

\subsection{Single image super-resolution}
\begin{figure*}[!ht]
    \centering
    \includegraphics[width=\textwidth]{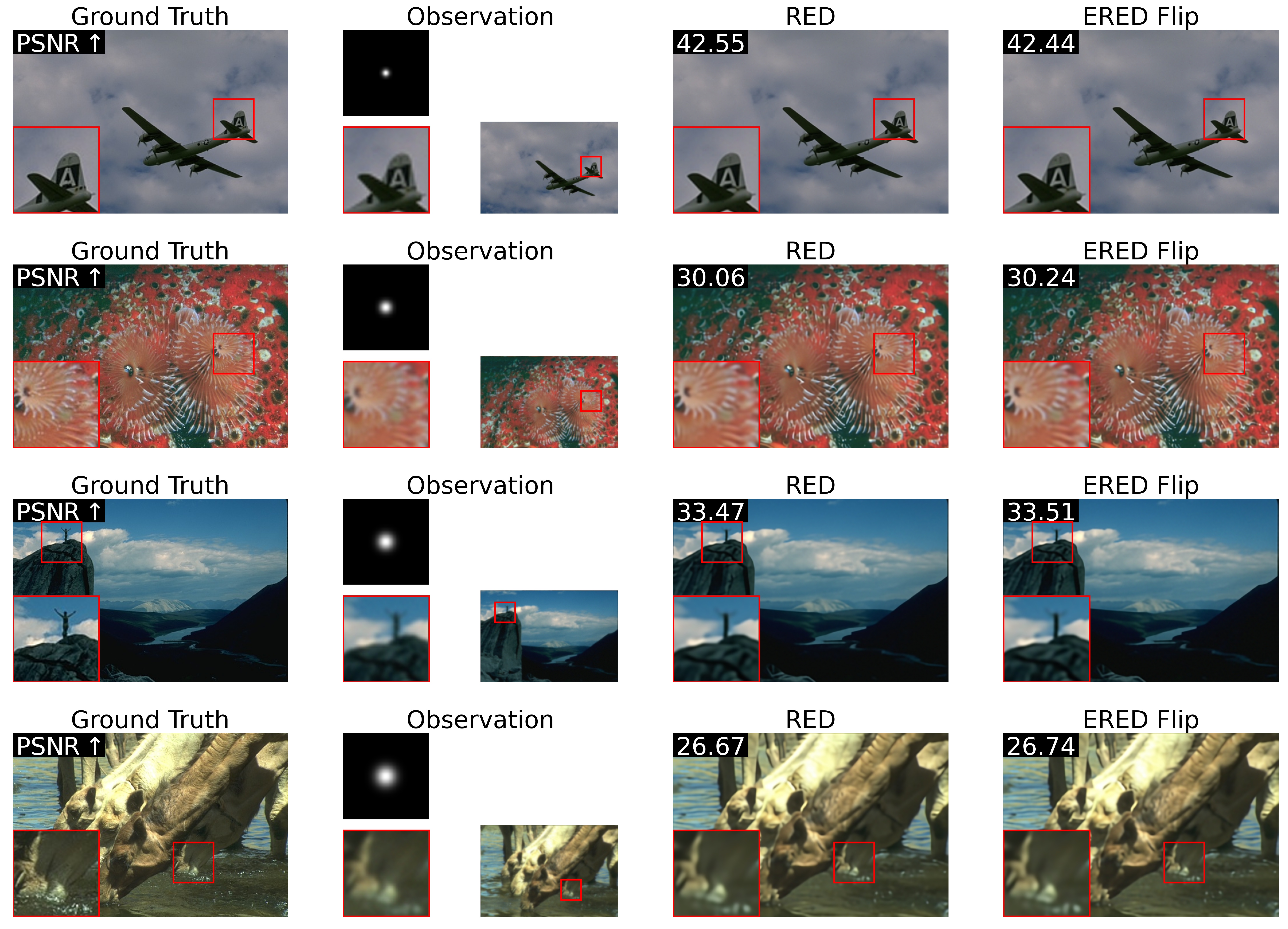}\vspace*{-0.2cm}
    \caption{Super-resolution with RED and ERED with a super-resolution factor of $2$ with a GS-denoiser trained on natural images. The set of transformation for ERED is random flip. Qualitative results of ERED and RED are very similar.  \vspace*{-0.2cm}
    }
    \label{fig:sr}
\end{figure*}

We now evaluate ERED on image super-resolution. 
In Table~\ref{tab:sr} and Figure~\ref{fig:sr}, we present the super-resolution results with a super-resolution factor of $2$ for RED and ERED on $10$ natural images extracted form the CBSD68 dataset. 
As suggested by the deblurring results of Table~\ref{table:quantitative_results}, we focus here on flip and rotation equivariance. 
We considered $8$ blur kernels including motion and fixed kernels, with a noise level of $\sigma_{y} = 1/255$.
In Table~\ref{tab:sr} and Figure~\ref{fig:sr}, we observe that the restoration performances are similar both quantitatively and qualitatively for RED and ERED.
This indicates that equivariance is not helpful to increase performance of image super-resolution.

\begin{table}[ht]
\centering
\begin{tabular}{c c c c}
Restoration method & PNSR $\uparrow$ & SSIM $\uparrow$ & N $\downarrow$\\
\hline
Bicubic & 25.47 & 0.72 & 200 \\
RED & 27.97 & 0.80 & 200\\
ERED Flip & 28.00 & 0.80 & 200\\
ERED Rotation & 28.01 & 0.80 & 200\\
\end{tabular}
\vspace*{0.2cm}
\caption{Super-resolution with a super-resolution factor of $2$ and $8$ different blur kernel (including fixed and motion blur) results on the CBSD10 ($10$ images from CBSD68) dataset with various restoration methods.\vspace*{-0.2cm}}
\label{tab:sr}
\end{table}

\subsection{Despeckling}

\begin{table}[ht]
\centering
\begin{tabular}{c c c c}
Restoration method & PNSR $\uparrow$ & SSIM $\uparrow$ & N $\downarrow$ \\
\hline
RED & 35.28 & 0.94 & 100\\
ERED Rotation & 35.19 & 0.94 & 100\\
ERED Flip & 35.59 & 0.94 & 100\\
\end{tabular}\vspace*{0.2cm}
\caption{Despeckling results on $60$ SAR images with various restoration methods. The number of looks is $L = 50$.}
\label{tab:despeckle}
\end{table}

In Table~\ref{tab:despeckle} and Figure~\ref{fig:despeckle}, we present the result of RED and ERED for Synthetic Aperture Radar (SAR) images despeckling. The speckle noise is multiplicative and implies that the data-fidelity is not $L$-smooth. Therefore, it is known to be more challenging than removing Gaussian noise. We use for this experiment the SAR dataset presented in~\cite{dalsasso2021if}. The test image (\textit{lely}) of this dataset has been cropped into $60$ images of size $256 \times 256$ to create a test dataset. The GS-DRUNet  has been retrained with the training images of this SAR dataset with the hyperparameters recommended in~\cite{hurault2022gradient} and flip data augmentation. 

The PSNR values provided in Table~\ref{tab:despeckle} show that random rotations degrade restoration performance, whereas flips are beneficial. This last observation is confirmed by the qualitative results of Figure~\ref{fig:despeckle}.
Therefore, flip equivariance appears to be advantageous for despeckling.
But more precisely, this confirms that it seems beneficial to use, at inference time, the same kind of equivariance that has been introduced by data augmentation at training time.

\begin{figure*}[!ht]
    \centering
    \includegraphics[width=\textwidth]{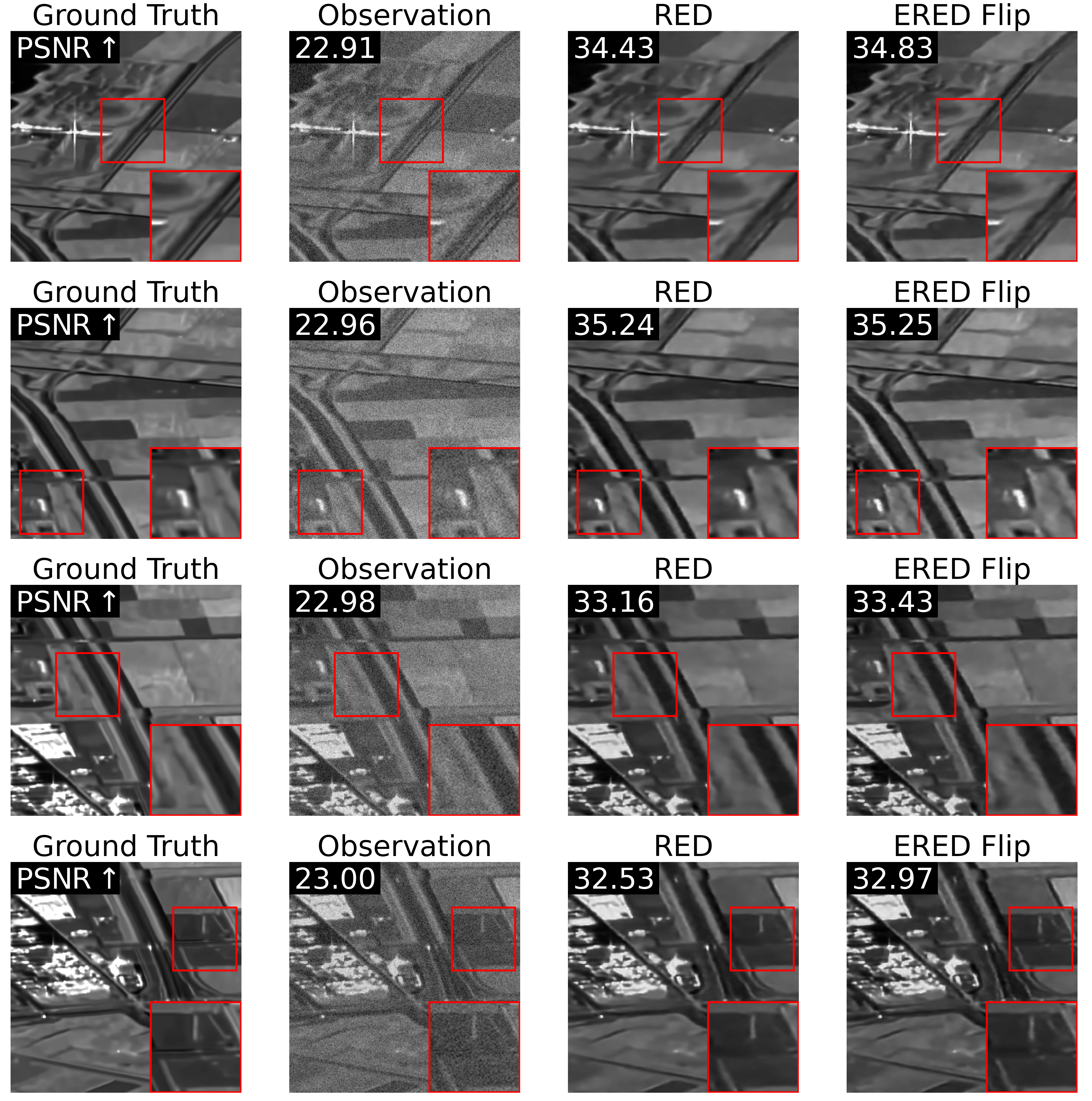}
    \caption{Despeckling with RED and ERED with a number of look of $L = 50$ with a GS-denoiser trained on SAR images. The set of transformation for ERED is random flip.  ERED produces a better qualitative result than RED.  \vspace*{-0.4cm}
    }
    \label{fig:despeckle}
\end{figure*}

\subsection{Denoising performance}
In order to understand the practical benefit of $\pi$-equivariance, we study the denoising performance of the equivariant denoiser as defined in Equation~\eqref{eq:equiv_denoiser}. 
In Table~\ref{table:denoising_performance}, we present the performance of each denoiser on natural images from the dataset CBSD68~\cite{martin2001} with various levels of noise. When the set of transformations is infinite, we take a Monte-Carlo approximation of Equation~\eqref{eq:equiv_denoiser} with $10$ random transformations. 
Surprisingly, the denoising performance with equivariance is similar to the one without equivariance. Performance with subpixel rotation denoiser is slightly lower, which suggests that the image distribution might not be $\pi$-equivariant to subpixel rotation. Observing that the denoising results are similar with all approaches suggests that the considered GS-DRUNet  has already learned these equivariance properties. 

\begin{table}[!ht]
\centering
\begin{tabular}{c c c c}
Denoising method & PNSR  & PNSR & PNSR\\
Noise level $\sigma$ &$5 / 255$&$10 / 255$&$20 / 255$\\
\hline
Simple & 40.54 & 36.46 & 32.73 \\
Rotation & 40.58 & 36.49 & 32.76 \\
Translation & 40.53 & 36.44 & 32.71\\ 
Subpixel Rotation & 40.34 & 36.26 & 32.56\\
Flip  & 40.58 & 36.49 & 32.76\\
\end{tabular}
\vspace*{0.2cm}
\caption{Denoising results on the CBSD68 dataset with various level of noise. \textit{Simple} denoising refers to an application of the GS-DRUNet denoiser~\cite{hurault2022gradient}, \textit{rotation} to the average of the denoising of the $4$ rotated images, \textit{flip} to the average of the denoising of the $4$ flip images, \textit{translation} to the average of the denoising of the $10$ random translated images and \textit{subpixel rotation} to the average of the denoising of the $10$ random subpixelic rotated images.}\label{table:denoising_performance}
\end{table}

\subsection{Equivariant RED with different denoisers}
Another interpretation of the improvement of ERED compared to RED is that the denoiser is applied on images close to its training domain, which consists of noisy images. In practice, data augmentation is one way to enforce favor the denoiser equivariance. For the GS-DRUNet and DnCNN architectures, training includes a random flip data augmentation but no random rotations. However, these two set of transformations are very close, e.g. a flip in both direction is equivalent to a rotation of angle $180^\circ$.

In Table~\ref{tab:deblurring_deblurring_various_denoisers}, we present the deblurring performance obtained with different denoisers. For GS-DRUNet, we use the pretrained weights released by~\cite{hurault2022gradient}. For DRUNet, we use the pretrained weights given in the Python librairie \textit{Deepinv}~\cite{tachella2025deepinverse}.  
For DnCNN, the noise of the observation is set to $\sigma_y = 1/255$ because the weights shared by~\cite{pesquet2021learning} were obtained after training with a single noise level $\sigma = 2/255$. 
Equivariance is slightly beneficial, with a gain of $+0.1$ to $+0.2$ dB. However, this interpretation in terms of being closer to the training domain does not explain why translation equivariance is also beneficial in the deblurring case, as observed in Table~\ref{table:quantitative_results}.

\begin{table}[!ht]
\centering
\begin{tabular}{c c c c c}
Denoiser & $\sigma_y$ & Method & PNSR $\uparrow$ & SSIM $\uparrow$ \\
\hline
\multirow{3}{*}{GS-DRUNet} & \multirow{3}{*}{$\frac{5}{255}$}  & RED & 29.84 & 0.84  \\
& & ERED rot. & 30.01 & 0.85 \\
& & ERED flip & 30.00 & 0.85 \\
\hline
\multirow{3}{*}{DRUNet} &\multirow{3}{*}{$\frac{5}{255}$} & RED & 29.24 & 0.81 \\
& & ERED rot. & 29.48 & 0.83  \\
& & ERED flip & 29.44 & 0.82 \\
\hline
\multirow{3}{*}{DnCNN} &\multirow{3}{*}{$\frac{1}{255}$} & RED & 35.26 & 0.94 \\
& & ERED rot. & 35.34 & 0.94 \\
& & ERED flip & 35.32 & 0.94 \\
\end{tabular}
\vspace*{0.2cm}
\caption{Deblurring results on CBSD10 ($10$ images extracted from CBSD68 dataset) with $10$ kernels of blur (including fixed and motion blur) with different pretrained denoisers. It is worth noting that the quantitative improvement with equivariance (ERED) is approximately $+0.2$ dB for each type of denoiser (rotation and flip).}
\label{tab:deblurring_deblurring_various_denoisers}
\end{table}

\section{Conclusion}
In this paper, we propose ERED, EPnP and SnoPnP, equivariant versions of RED and PnP. We provide an interpretation of the ERED algorithm as an equivariant property of the underlying prior. We give theoretical convergence results for all the proposed algorithms and a critical point convergence with an equivariant prior $p$ for ERED (Proposition~\ref{prop:critical_points_cvg}). We refine the convergence results obtain for EPnP in the particular case of SnoPnP.  Experimental results illustrate the modest improvement brought by such methods. Our experiments support the heuristic that using the same equivariance at inference time than at training time with data augmentation is beneficial.

\section*{Acknowledgements}
This study has been carried out with financial support from the French Direction G\'en\'erale de l’Armement. Experiments presented in this paper were carried out using the PlaFRIM experimental testbed, 
supported by Inria, CNRS (LABRI and IMB), Universite de Bordeaux, Bordeaux INP and Conseil Regional d’Aquitaine (see https://www.plafrim.fr).

\bibliography{ref}

\clearpage
\onecolumn
\appendix

\section{Hyper parameters for the experiments}\label{sec:hyperparameter}
In this section, we detail the hyper parameters choice for our experiments. Grid searches have been made to find the optimal parameters in term on PSNR for each method.
For image deblurring, Annealed SNORE parameters have been chosen according to the recommendation in~\cite{renaud2024plug}.

\begin{table}[h]
    \centering
    \begin{tabular}{c c c c c c c c}
        Problem & Noise level & Denoiser & Method & $\delta$ & $\sigma$ & $\lambda$ &  $N$ \\
        \hline
        \multirow{8}{*}{Deblurring} & \multirow{3}{*}{GS-DRUNet} & \multirow{3}{*}{$\sigma_y = 5/255$} & RED & $1.5$ & $7/255$ & $0.15$ & $400$ \\
        & & & ERED & $1.5$ & $8/255$ & $0.17$ & $400$ \\
        & & & SNORE & $1.5$ & $5/255$ & $0.5$ & $1000$ \\
        & & & PnP & $1$ & $4/255$ & $0.53$ & $400$ \\
        & & & EPnP & $1$ & $4/255$ & $0.53$ & $400$ \\
        & & & SnoPnP & $1$ & $5/255$ &$0.53$ & $100$ \\
        \cline{2-8}
        & \multirow{3}{*}{DRUNet} & \multirow{3}{*}{$\sigma_y = 5/255$}& RED & $1.5$ & $9/255$ & $0.12$ & $400$ \\
        & & & ERED flip & $1.5$ & $8/255$ & $0.14$ & $400$ \\
        & & & ERED rotation & $1.5$ & $8/255$ & $0.15$ & $400$ \\
        \cline{2-8}
        & \multirow{2}{*}{DnCNN} & \multirow{2}{*}{$\sigma_y = 1/255$} & RED & $2.0$ & $2/255$ & $0.11$ & $400$ \\
        & & & ERED & $2.0$ & $2/255$ & $0.11$ & $400$ \\
        \cline{2-8}
        & \multirow{1}{*}{Prox-DRUNet} & \multirow{1}{*}{$\sigma_y = 5/255$} & SnoPnP & $1$ & $5/255$ & $1.05$ & $1500$ \\

        \hline
        \multirow{2}{*}{Super-resolution} & \multirow{2}{*}{GS-DRUNet} & \multirow{2}{*}{$\sigma_y = 1/255$} & RED & $2.0$ & $11/255$ & $0.07$ & $200$ \\
        & & & ERED & $2.0$ & $13/255$ & $0.05$ & $200$ \\
        \hline
        \multirow{2}{*}{Despeckling} & \multirow{2}{*}{GS-DRUNet} & \multirow{2}{*}{$L = 50$} & RED & $0.01$ & $8/255$ & $100$ & $100$ \\
        & & & ERED & $0.01$ & $8/255$ & $100$ & $100$ \\
    \end{tabular}
    \caption{Hyper parameters setting for various experiments and methods.}
\end{table}

\section{Definition of subanalytic functions}\label{sec:def_subanalytic}
In this part, we recall the geometrical definition for the notion of subanalytic functions.
\begin{definition}(\cite{bierstone1988semianalytic, bolte2007lojasiewicz})
\begin{itemize}
    \item A subset $S$ of $\R^d$ is a semianalytic set if each point of $\R^d$ admits a neighborhood $V$ for which there exists a finite number of real analytic functions, i.e. equal locally to a power series, $f_{i,j}, g_{i,j} : \R^d \to \R$ such that 
    \begin{align}
        S \cap V = \bigcup_{j=1}^p \bigcap_{i=1}^q \{x \in \R^d | f_{i,j}(x) = 0, g_{i,j}(x) < 0\}.
    \end{align}
    
    \item A subset $S$ of $\R^d$ is a subanalytic set if each point of $\R^d$ amits a neighborhood $V$ for which
    \begin{align}
        S \cap V = \{x \in \R^d | (x,y) \in U\},
    \end{align}
    where $U$ is a bounded semianalytic subset of $\R^d \times \R^m$ for some $m \ge 1$.
    
    \item A function $f:\R^d \to \R \cup \{+\infty\}$ is called subanalytic if its graph $\{(x,y) \in \R^d \times \R | y = f(x)\}$ is a subanalytic subset of $\R^d \times \R$.
\end{itemize}
\end{definition}
    
All compositions of standard functions (polynomial, exponential, logarithmic) are subanalytic. Therefore, if the forward model is expressed with standard functions, then the data-fidelity $f$ is subanalytic. This is the case for linear forward model with a model of noise expressed with standard functions, e.g. Gaussian, Poisson, Fisher-Tippett.
    
By \cite[Proposition 5]{hurault2024convergent}, in the case of a gradient step denoiser $D_{\sigma} = Id - \nabla h_{\sigma}$ with $\nabla h_{\sigma}$ $L_g$-Lipschitz and $L_g < 1$, we have $D_{\sigma} = \mathsf{Prox}_{g_{\sigma}}$ and an expressed of $g_{\sigma}$ on $\mathsf{Im}(D_{\sigma})$, there exists $K \in \R$ such that $\forall x \in \mathsf{Im}(D_{\sigma})$, we have
\begin{align}
    g_{\sigma}(x) = h_{\sigma}(D_{\sigma}^{-1}(x)) - \frac{1}{2} \|D_{\sigma}^{-1}(x) - x\|^2 + K.
\end{align}
Moreover, $\nabla D_{\sigma} = I_d - \nabla^2 h_{\sigma} > 0$ is non-singular, because $\nabla h_{\sigma}$ is $L_g$-Lipschitz with $L_g < 1$. Therefore, if $D_{\sigma}$ is real analytic, i.e. its activation functions are real analytic, then by the real analytic inverse function theorem, the function $D_{\sigma}^{-1}$ is then real analytic on $\mathsf{Im}(D_{\sigma})$. Then by sum and composition of real analytic functions, $g_{\sigma}$ is real analytic on $\mathsf{Im}(D_{\sigma})$.
    
Therefore, in this context, $F = f + \lambda g_{\sigma}$ is real analytic on $\mathsf{Im}(D_{\sigma})$.
However, it is not possible to verify that $\fF$ is subanalytic on all $\R^d$ because $h_{\sigma}$ is not defined explicitly outside $\mathsf{Im}(D_{\sigma})$.

\section{Technical proofs}

\subsection{Preliminary technical results}
\subsubsection{An inequality for weakly convex function}
First, we state a simple technical result that will be useful in our proofs.
\begin{lemma}\label{lemma:inequality_weakly_cvx}
For $g : \R^d \to \R$ $\rho$-weakly convex and differentiable, we have that $\forall x, y \in \R^d$
\begin{align*}
    \langle \nabla g(x), y - x \rangle \le g(y) - g(x) +\frac{\rho}{2} \|y - x\|^2. 
\end{align*}
\end{lemma}
\begin{proof}
By definition of the weak convexity, $g +\frac{\rho}{2}\|\cdot\|^2$ is convex. Therefore, we have
\begin{align*}
    g(x) +\frac{\rho}{2}\|x\|^2 + \langle \nabla g(x) + \rho x, y - x \rangle \le g(y) +\frac{\rho}{2}\|y\|^2 \\
    \langle \nabla g(x) + \rho x, y - x \rangle \le g(y) - g(x) +\frac{\rho}{2} \|y - x\|^2.
\end{align*}
\end{proof}

\subsection{Useful biased optimization result}\label{app:thm}
For completeness of the paper, here we recall Theorem 2.1 (ii) from~\cite{tadic2017asymptotic}. This theorem tackles the convergence of a sequence $x_k$ defined by a biased stochastic gradient descent algorithm, i.e. there exists $f:\R^d \to \R^d$ differentiable, such that 
$$x_{k+1} = x_k - \delta_k (\nabla f(x_k) + \xi_k), $$
with $\delta_k > 0$ the step-size and $\xi_k$ the bias noise.

\begin{assumption}\label{ass:tadic}
    \textbf{(a)} $\lim_{k\to +\infty}\delta_k = 0$ and $\sum_{k=0}^{+\infty} \delta_k = +\infty$. 
    
    \textbf{(b)} $\xi_k$ admits the decomposition $\xi_k = \zeta_k + \eta_k$  that satisfies, for all $k \ge 0$ and almost surely on $\{\sup_{k \in \N}{\|x_k\|} < + \infty\}$:
    $$ \lim_{k\to +\infty} \max_{k\le n < a(k,t)} \| \sum_{i=k}^n \delta_k \zeta_k \| = 0, \limsup_{k \to +\infty}{\|\eta_k\|} < + \infty,$$
    where $a(k,t)$ is defined for $t> 0$ by $a(k,t) = \max \{  n \le k | \sum_{i=k}^{n-1} \delta_k \le t \}$.

    \textbf{(c)} $f$ is $p$-times differentiable on $\R^d$ with $p > d$.
\end{assumption}
    
\begin{theorem}{\cite{tadic2017asymptotic}}\label{theorem:tadic}
    Under Assumption~\ref{ass:tadic}, for a compact $Q \subset \R^d$ there exists a real number $K_{Q} > 0$ (depending only of $f$) such that it holds almost surely on $\lambda_Q = \{x_k \in Q | \forall k \in \N\}$  that
    \begin{equation}
        \limsup_{k \to + \infty} \|\nabla f(x_k)\| \le K \eta^{\frac{q}{2}},~~\limsup_{k \to + \infty} f(x_k) - \liminf_{k \to + \infty} f(x_k) \le K \eta^{q},
    \end{equation}
     with $q = \frac{p-d}{p-1}$ and $\eta = \limsup_{k \to +\infty} \|\eta_k\|$.
\end{theorem}

\subsection{Proof of Proposition~\ref{prop:haar}}\label{sec:proof_haar}
For $x \in \R^d$, by the right-invariance of $\pi$, we get
\begin{align*}
    &\eE_{G' \sim \pi}\left( \log(r_{\sigma}^{\pi}(G'(x)) \right)= \int_{\G} \log(r_{\sigma}^{\pi}(G'(x)) d\pi(G') \\
    =& \int_{\G} \log\left( -\int_{\G}  \log (p_{\sigma} (G \circ G'(x)) d\pi(G) \right) d\pi(G') \\
    = &\int_{\G} \log\left( -\int_{\G}  \log (p_{\sigma} (G(x)) d\pi(G) \right) d\pi(G') 
     \\= &\log\left( r_{\sigma}^{\pi}(x) \right).
\end{align*}

\subsection{Proof of Lemma~\ref{lemma:bounded_variance_1}}\label{sec:proof_bounded_variance_1}

By Assumption~\ref{ass:prior_score_approx}-\ref{ass:finite_moment}, and the inequality $\forall x,y \in \R^+, (x+y)^2 \le 2(x^2+y^2)$, we have
\begin{align*}
  &\eE(\|\xi_k\|^2 | x_k)\\
   =& \lambda^2 \eE(\|  \fJ_{G}^T(x_k) \nabla \log p_{\sigma}(G(x_k)) - \eE_{G \sim \pi}\left(  \fJ_{G}^T(x_k) \nabla \log p_{\sigma}(G(x_k)) \right)\|^2 | x_k) \\
   =& \lambda^2 \eE\left(\|  \fJ_{G}^T(x_k) \nabla \log p_{\sigma}(G(x_k))\|^2| x_k\right) - \lambda^2 \|\eE_{G \sim \pi}\left(  \fJ_{G}^T(x_k) \nabla \log p_{\sigma}(G(x_k)) | x_k \right)\|^2  \\
    \le &\lambda^2 \eE(\|  \fJ_{G}^T(x_k) \nabla \log p_{\sigma}(G(x_k))\|^2 | x_k) \\
    \le& \lambda^2 \eE(\vvvert  \fJ_{G}^T(x_k) \vvvert^2 \| \nabla \log p_{\sigma}(G(x_k))\|^2 | x_k) \\
    \le& \lambda^2 \eE(\vvvert  \fJ_{G}(x_k) \vvvert^2 B^2 \sigma^{2 \beta} \left( 1 + \|G(x_k)\|^{n_1} \right)^2 | x_k) \\
    \le& 2 \lambda^2 B^2 \sigma^{2 \beta} \eE(\vvvert  \fJ_{G}(x_k)\vvvert^2 \left( 1 + \|G(x_k)\|^{2n_1} \right) | x_k) \\
    \le& 2 \lambda^2 B^2 \sigma^{2 \beta} \left( \eE(\vvvert  \fJ_{G}(x_k)\vvvert^2| x_k) + \eE(\vvvert  \fJ_{G}(x_k)\vvvert^2 \|G(x_k)\|^{2n_1} | x_k)  \right),
\end{align*}
where $n_1$ is defined in Assumption~\ref{ass:prior_score_approx}(c).
By using Young inequality, i.e. $\forall x, y \in\R^+, p, q > 1$ such that $\frac{1}{p}+ \frac{1}{q} = 1$, $|xy|\le \frac{x^p}{p} + \frac{y^q}{q}$, with $p= 1+\frac{\epsilon}{2}$, we get for $p = \frac{m_{n_1, \epsilon}}{\frac{2n_1(2+\epsilon)}{\epsilon}}$ with $m_{n_1, \epsilon} = \lceil \frac{2n_1(2+\epsilon)}{\epsilon} \rceil$
\begin{align*}
&\eE(\|\xi_k\|^2 | x_k) \\
    \le& 2 \lambda^2 B^2 \sigma^{2 \beta} \left( \frac{4}{2+\epsilon} \eE(\vvvert  \fJ_{G}(x_k)\vvvert^{2+\epsilon}| x_k) + \frac{\epsilon}{2+\epsilon}  + \frac{\epsilon}{2+\epsilon}\eE( \|G(x_k)\|^{\frac{2n_1(2+\epsilon)}{\epsilon}} | x_k)  \right)\\
    \le& 2 \lambda^2 B^2 \sigma^{2 \beta} \Bigl( \frac{4M_{2+\epsilon}+\epsilon}{2+\epsilon}   + \frac{1}{m_{n_1,\epsilon}} \Bigl( 2n_1 \eE( \|G(x_k)\|^{m_{n_1,\epsilon}} | x_k) + m_{n_1,\epsilon}- \frac{2n_1(2+\epsilon)}{\epsilon} \Bigr) \Bigr)\\
    \le& 2 \lambda^2 B^2 \sigma^{2 \beta} \left( \frac{4M_{2+\epsilon}+\epsilon}{2+\epsilon}   + \frac{1}{m_{n_1,\epsilon}} \left( 2n_1 C_{\K, m_{n_1, \epsilon}} + m_{n_1,\epsilon} - \frac{2n_1(2+\epsilon)}{\epsilon} \right) \right) := C ,
\end{align*}
with $C< +\infty$ a constant independent of $k$ and $x_k$. This proves Lemma~\ref{lemma:bounded_variance_1} by taking the expectation on $x_k$ and using the law of total expectation.

\subsection{Proof of Proposition~\ref{prop:convergence_unbiased}}\label{sec:proof_convergence_unbiased}
The proof is obtained by applying  \cite[Theorem 2.1, (ii)] {tadic2017asymptotic}, which is recalled in Appendix~\ref{app:thm} (Theorem~\ref{theorem:tadic}) for completeness. To do so, we have to verify the different assumptions (Assumption~\ref{ass:tadic} in Appendix~\ref{app:thm}) of this theorem. 
First, Assumption~\ref{ass:tadic}(a) is verified by Assumption~\ref{ass:step_size_decreas}(a). 
Next $p_{\sigma}$ is $\mathcal{C}^{\infty}$ by convolution with a Gaussian. Then $\log p_{\sigma}$ is also $\mathcal{C}^{\infty}$~\cite{laumont2023maximum} and so is $r_{\sigma}^{\pi}$ defined in relation~\eqref{eq:eq_reg}. By Assumption~\ref{ass:data_fidelity_reg}(b), $\mathcal{F}_{\sigma}^{\pi} = f + \lambda r_{\sigma}^{\pi}$ is $\mathcal{C}^{\infty}$. So Assumption~\ref{ass:tadic}(c) is verified.

Now, we are going to prove that Assumption~\ref{ass:tadic}(b) is verified, i.e. the noise fluctuation can be controlled.
To that end, we define 
\begin{align*}
\xi_k &= \nabla f(x_k) + \lambda \fJ_{G}^T(x_k) \nabla \log p_{\sigma}(G (x_k)) - \nabla \fF_{\sigma}^{\pi} (x_k)\\ &= \lambda \left( \fJ_{G}^T(x_k) \nabla \log p_{\sigma}(G (x_k)) \right.\\
&\left.\;\;\;\;\;\;\;\;- \eE_{G \sim \pi}\left(\fJ_{G}^T(x_k) \nabla \log p_{\sigma}(G(x_k)) \right) \right),\end{align*} 
with $\lambda >0, G \sim \pi$. By definition we have $\eE(\xi_k) = 0$ and, from~\eqref{eq:theoretical_process}, we get
\begin{align*}
    x_{k+1} = x_k - \delta_k (\nabla \fF_{\sigma}^{\pi}(x_k) + \xi_k).
\end{align*}

\begin{lemma}\label{lemma:bounded_variance_1}
    Under Assumptions~\ref{ass:prior_score_approx}-\ref{ass:finite_moment}, almost surely on $\Lambda_{\K}$, there exists $C \hspace{-1.5pt}> \hspace{-1.5pt}0$ such that $\forall k\hspace{-1.5pt} \in \hspace{-1.5pt}\N$, \hspace{-1.5pt}$\eE(\|\xi_k\|^2) \hspace{-1.5pt}\le \hspace{-1.5pt}C$.
\end{lemma}

Lemma~\ref{lemma:bounded_variance_1} is demonstrated in Appendix~\ref{sec:proof_bounded_variance_1}. By using Lemma~\ref{lemma:bounded_variance_1}, we get $\sum_{k \in \N}{\delta_k^2 \eE(\|\xi_k\|^2)} \le C \sum_{k \in \N}{\delta_k^2} < + \infty$, by Assumption~\ref{ass:step_size_decreas}. Then, we deduce from the Doob inequality (which holds because $\sum_{k=n}^l\xi_k$ is a martingale) that
\begin{align}\label{eq:doob_inequality}
    \eE\left(\sup_{n \le l \le m}{\| \sum_{k=n}^l{\delta_k \xi_k} \|^2}\right) \le 4 \sum_{k = n}^m {\delta_k^2 \eE\left(\|\xi_k\|^2\right)}.
\end{align}
By the monotone convergence theorem, this implies
\begin{align*}
    \eE\left(\sup_{n \le l}{\| \sum_{k=n}^l{\delta_k \xi_k} \|^2}\right) &\le 4 \sum_{k = n}^{\infty} {\delta_k^2 \eE\left(\|\xi_k\|^2\right)} \\&
    \leq 4 C \sum_{k = n}^{\infty} \delta_k^2 .
\end{align*}
Thus the sequence $(\sup_{n \le l}{\| \sum_{k=n}^l{\delta_k \xi_k} \|^2})_n$ tends to zero in $L^1$ and also almost surely (because it is non-increasing). The square function being non-decreasing, it implies that $(\sup_{n \le l}{\| \sum_{k=n}^l{\delta_k \xi_k} \|})_n$ tends to zero almost surely.
The process~\eqref{eq:theoretical_process} thus verifies Assumption~\ref{ass:tadic}(b) almost surely. 
We can apply Theorem~\ref{theorem:tadic}, which concludes the proof.

\subsection{Proof of Lemma~\ref{lemma:bounded_variance_2}}\label{sec:proof_bounded_variance_2}
Using the definition of $\xi_k$, Assumptions~\ref{ass:finite_moment}-\ref{ass:denoiser_sub_polynomial}, and Young inequality, we have, almost surely on $\Lambda_{\K}$,
\begin{align*}
    &\eE(\|\gamma_k\|^2|x_k) = \eE(\|\xi_k - \eE(\xi_k)\|^2|x_k) \\
    \le& \frac{\lambda^2}{\sigma^4} \eE\left( \| \fJ_{G}^T(x_k) \left( G(x_k) - D_{\sigma}(G(x_k))\right) \|^2|x_k \right) \\
    \le& \frac{\lambda^2}{\sigma^4} \eE\left( \vvvert\fJ_{G}(x_k) \vvvert^2 \|  G(x_k) - D_{\sigma}(G(x_k)) \|^2|x_k \right) \\
    \le &\frac{2\lambda^2}{\sigma^4} \eE\left( \vvvert\fJ_{G}(x_k) \vvvert^2 \left( \|G(x_k)\|^2 + \|D_{\sigma}(G(x_k)) \|^2 \right)|x_k \right) 
    \\\le &\frac{2\lambda^2}{\sigma^4} \eE\left( \vvvert\fJ_{G}(x_k) \vvvert^2 \left( \|G(x_k)\|^2 + 2 C^2 (1 + \|G(x_k) \|^{2n_2} \right)|x_k \right) \\
    \le& \frac{2\lambda^2}{\sigma^4} \Bigl( 2 C^2 \eE\left(\vvvert\fJ_{G}(x_k) \vvvert^2\right) + \eE\left(\vvvert\fJ_{G}(x_k) \vvvert^2 \|G(x_k)\|^2 \right)\\&+ 2 C^2 \eE\left(\vvvert\fJ_{G}(x_k) \vvvert^2 \|G(x_k)\|^{2n_2} \right)  \Bigr) \\
    \le &\frac{2\lambda^2}{\sigma^4} \bigg( \frac{4 C^2}{2+\epsilon} \eE\left(\vvvert\fJ_{G}(x_k)\vvvert^{2+\epsilon}\right) +\frac{\epsilon}{2+\epsilon} + \frac{2}{2+\epsilon} \eE\left(\vvvert\fJ_{G}(x_k) \vvvert^{2+\epsilon}\right)  \\
    & +\frac{\epsilon}{2+\epsilon}\eE\left( \|G(x_k)\|^{\frac{2(2+\epsilon)}{\epsilon}} \right)+ \frac{4 C^2}{2+\epsilon} \eE\left(\vvvert\fJ_{G}(x_k) \vvvert^{2+\epsilon} \right)\\& + \frac{2C^2\epsilon}{2+\epsilon} \eE\left( \|G(x_k)\|^{\frac{2n_2(2+\epsilon)}{\epsilon}} \right)\hspace{-0.1cm}\bigg) \\
    \le &\frac{2\lambda^2}{\sigma^4} \bigg( \frac{(8 C^2 + 2)M_{2+\epsilon}+\epsilon}{2+\epsilon}  +\frac{2}{\lceil c_{\epsilon} \rceil}\eE\left( \|G(x_k)\|^{\lceil c_{\epsilon} \rceil} \right) + \frac{\lceil c_{\epsilon} \rceil - c_{\epsilon}}{\lceil c_{\epsilon} \rceil} \\&+ \frac{4 n_2 C^2}{\lceil n_2 c_{\epsilon} \rceil} \eE\left( \|G(x_k)\|^{\lceil n_2 c_{\epsilon} \rceil} \right) + \frac{\lceil n_2 c_{\epsilon} \rceil - n_2 c_{\epsilon}}{\lceil n_2 c_{\epsilon} \rceil}  \bigg) \\
    \le &\frac{2\lambda^2}{\sigma^4} \hspace{-1pt}\bigg( \hspace{-1pt}\frac{(8 C^2 + 2)M_{2+\epsilon}+\epsilon}{2+\epsilon} \hspace{-1pt}  +\hspace{-1pt}\left(\frac{2}{\lceil c_{\epsilon} \rceil}\hspace{-1pt}+ \hspace{-1pt}\frac{4 n_2 C^2}{\lceil n_2 c_{\epsilon} \rceil}\right)\hspace{-1pt}C_{\K, \lceil c_{\epsilon} \rceil}\hspace{-1pt}+ \hspace{-1pt}\frac{\lceil c_{\epsilon} \rceil - c_{\epsilon}}{\lceil c_{\epsilon} \rceil} \hspace{-1pt}\\&+ \hspace{-1pt}\frac{\lceil n_2 c_{\epsilon} \rceil - n_2 c_{\epsilon}}{\lceil n_2 c_{\epsilon} \rceil}  \bigg) \\:= &C_2 < +\infty,
\end{align*}
with $c_{\epsilon} = \frac{2(2+\epsilon)}{\epsilon}$ and $n_2$ defined in Assumption~\ref{ass:denoiser_sub_polynomial}. This proves Lemma~\ref{lemma:bounded_variance_2} by taking the expectation on $x_k$.

\subsection{Proof of Proposition~\ref{prop:convergence_biaised}}\label{sec:proof_convergence_biaised}
First the bias is denoted by $\eta_k = \eE(\xi_k)$ and the noise by ${\gamma_k = \xi_k - \eE(\xi_k)}$. So we have $\xi_k = \gamma_k + \eta_k$ and $\eE(\gamma_k) = 0$.
We apply again~Theorem~\ref{theorem:tadic}. Assumptions~\ref{ass:tadic}(a)-\ref{ass:tadic}(c) are verified thanks to Assumptions~\ref{ass:step_size_decreas}(a) and~\ref{ass:data_fidelity_reg}(b).

\begin{lemma}\label{lemma:bounded_variance_2}
    Under Assumptions~\ref{ass:finite_moment}-\ref{ass:denoiser_sub_polynomial}, almost surely on $\Lambda_{\K}$, there exists $C_2 \hspace{-1.5pt}> \hspace{-1.5pt}0$ such that $\forall k \hspace{-1.5pt}\in \hspace{-1.5pt}\N,$ \hspace{-1.5pt}$\eE(\|\xi_k\|^2) \hspace{-1.5pt}\le\hspace{-1.5pt} C_2$.
\end{lemma}

Lemma~\ref{lemma:bounded_variance_2} is proved in Section~\ref{sec:proof_bounded_variance_2}. Then, by using Lemma~\ref{lemma:bounded_variance_2} and the Doob inequality as in relation~\eqref{eq:doob_inequality}, we obtain Assumption~\ref{ass:tadic}(b) of \cite{tadic2017asymptotic} (i.e. the noise fluctuation is controlled) and we can apply Theorem~\ref{theorem:tadic} to obtain equations~\eqref{eq1}-\eqref{eq2}.
Under Assumption~\ref{ass:g_bounded}, we study the asymptotic behavior of $\eta_k$,
\begin{align*}
    &\|\eta_k\| \\
    =&\; \|\eE(\xi_k)\| = \frac{\lambda}{\sigma^2}\left\|\eE\left( \fJ_G^T(x_k) \left(D_{\sigma} - D^{\ast}_{\sigma}\right)(G(x_k)) \right)\right\| \\
    \le&\; \frac{\lambda}{\sigma^2} \eE\left( \vvvert \fJ_G(x_k) \vvvert \|\left(D_{\sigma} - D^{\ast}_{\sigma}\right)(G(x_k))\| \right).
\end{align*}
By Assumption~\ref{ass:g_bounded}, because $x_k \in \K$, we know that $G(x_k) \in \L = \mathcal{B}(0, C_{\K})$. So we have $\|\left(D_{\sigma} - D^{\ast}_{\sigma}\right)(G(x_k))\| \le \|D_{\sigma} - D^{\ast}_{\sigma}\|_{\infty, \L}$ and the desired inequality~\eqref{ineq}.

\subsection{Proof of Proposition~\ref{prop:uniform_cvg}}\label{sec:proof_uniform_cvg}
Due to the $\pi$-equivariance of $p$, we have $s = \eE_{G \sim \pi}\left(\fJ_{G}^T (s \circ G)  \right)$ with a random variable $G \sim \pi$. With the definition of $s_{\sigma}^{\pi}$~\eqref{eq:equiv_score}, we get for $x \in \R^d$
\begin{align}\nonumber
    &s_{\sigma}^{\pi}(x) - \nabla \log p(x) \\=& \eE_{G \sim \pi} \left(\fJ_G^T(x) \left(\nabla \log p_{\sigma} - \nabla \log p \right)(G(x))\right).
\end{align}
By Assumption~\ref{ass:g_bounded}, we get that $\forall x \in \K, G(x) \in \mathcal{B}(0, C_{\K})$, with $\L = \mathcal{B}(0, C_{\K})$ the closed ball of center $0$ and radius $C_{\K}$. By Assumption~\ref{ass:p_regularity}(a) and Proposition 1 in~\cite{laumont2023maximum}, we know that $\|\nabla \log p_{\sigma} - \nabla \log p\|_{\infty, \L} \to 0$ when $\sigma \to 0$.
Then,  when $\sigma \to 0$, using Assumption~\ref{ass:g_bounded}, we obtain
\begin{align*}
    &\|s - s_{\sigma}^{\pi}\|_{\infty, \K} \\
    \le&\; \eE_{G \sim \pi} \left(\vvvert\fJ_G\vvvert \|\nabla \log p_{\sigma} - \nabla \log p \|_{\infty, \L}\right) \\
    \le&\; \|\nabla \log p_{\sigma} - \nabla \log p \|_{\infty, \L} \eE_{G \sim \pi}\left(\vvvert\fJ_G\vvvert \right) \to 0. 
\end{align*}

\subsection{Proof of Proposition~\ref{prop:critical_points_cvg}}\label{sec:proof_critical_points_cvg}
For $x \in \S$, we have $\sigma_n >0$ decreasing to $0$ and $x_n \in \S_{\sigma_n}$ such that $x_n \to x$. Because $x_n$ is a converging sequence, there exists a compact $\K$ such that $\forall n \in \N, x_n \in \K$. Moreover, thanks to Proposition~\ref{prop:uniform_cvg}, $\|s - s_{\sigma_n}^{\G}\|_{\infty, \K} \xrightarrow[n \to \infty]{} 0$ and then $\|\fF - \fF_{\sigma_n}^{\pi}\|_{\infty, \K} \to 0$. So, $\|\fF(x_n) - \fF_{\sigma_n}^{\pi}(x_n)\| \to 0$ and by definition $\fF_{\sigma_n}^{\pi}(x_n) = 0$ which gives $\|\fF(x_n)\| \to 0$. It implies that $\|\nabla \fF(x)\| = 0$.

\subsection{Proof of Lemma~\ref{lemma:residual_control_perturbed_pgd}}\label{sec:proox_residual_control_perturbed_pgd}

Thanks to the $L_f$-smoothness of $f$ (Assumption~\ref{ass:regularities_assumptions}) on all $\R^d$, we have
\begin{align}\label{eq:f_l_smooth}
    &f(x_{k+1} - \zeta_{k+1}) \le f(x_k - \zeta_k) + \langle \nabla f(x_k- \zeta_k), x_{k+1} - x_k - \zeta_{k+1} + \zeta_k \rangle \\
    &+ \frac{L_f}{2} \|x_{k+1} - x_k - \zeta_{k+1} + \zeta_k\|^2 \nonumber
\end{align}
We now exploit Equation~\eqref{eq:perturbe_version_snopnp} that can be expressed as $x_{k+1} - \zeta_{k+1} =  \mathsf{Prox}_{\tilde g_{\sigma}}\left( x_k - \frac{1}{\lambda} \nabla f(x_{k}) \right)$. It is important to note that $x_{k+1} - \zeta_{k+1}$ is a deterministic function of $x_k$. We denote $\eE_k = \eE(\cdot|x_k)$. By taking the expectation in Equation~\eqref{eq:f_l_smooth}, because $x_k-\zeta_k$ is a deterministic function of $x_{k-1}$ and $\zeta_{k}$, $\zeta_{k+1}$ are centered and independent of $x_k-\zeta_k$, we have $\eE\left(\langle \nabla f(x_k- \zeta_k), \zeta_k - \zeta_{k+1}  \rangle \right) = 0$. Therefore, we get
\begin{align}
    &\eE\left(f(x_{k+1} - \zeta_{k+1})\right) \le \eE\left(f(x_k - \zeta_k)\right) + \eE\left(\langle \nabla f(x_k), x_{k+1} - x_k \rangle\right) + L_f \eE\left(\|x_{k+1} - x_k\|^2\right) \label{eq:ineq_inter_lemma}\\
    &+ L_f \eE\left(\| \zeta_{k+1} - \zeta_k\|^2\right)\nonumber.
\end{align}
Then, the optimal condition of the proximal operator in Equation~\eqref{eq:perturbe_version_snopnp} gives 
\begin{align}
    0&=x_{k+1} -  \zeta_{k+1} - x_k + \frac{1}{\lambda} \nabla f(x_k) + \nabla \tilde g_{\sigma}(x_{k+1} -  \zeta_{k+1})  \nonumber \\
    \nabla f(x_k) &= \lambda (x_k - x_{k+1}) + \lambda \zeta_{k+1} - \lambda \nabla \tilde g_{\sigma}(x_{k+1} -  \zeta_{k+1}). \label{eq:formulation_nabla_f}
\end{align}
Injecting Equation~\eqref{eq:formulation_nabla_f} into Equation~\eqref{eq:ineq_inter_lemma}, we get
\begin{align}
    &\eE\left(f(x_{k+1} - \zeta_{k+1})\right) \le \eE\left(f(x_k - \zeta_k)\right) + \left(L_f - \lambda\right) \eE\left(\|x_{k+1} - x_k\|^2\right) + \lambda \eE\left(\langle \zeta_{k+1}, x_{k+1} - x_k \rangle\right) \nonumber \\ &- \lambda \eE\left(\langle \nabla \tilde g_{\sigma}(x_{k+1} -  \zeta_{k+1}), x_{k+1} - x_k \rangle\right) \nonumber\\ &+ \eE\left(\langle \nabla f(x_k- \zeta_k) - \nabla f(x_k), x_{k+1} - x_k\rangle\right)  + L_f \eE\left(\| \zeta_{k+1} - \zeta_k\|^2\right). \label{eq:intermediar_bound} 
\end{align}
Thanks to Equations~\eqref{eq:formulation_nabla_f} and~\eqref{eq:variance_zeta}, we have 
\begin{align}\label{eq:control_noise_intermediar}
    &\eE(\langle \zeta_{k+1}, x_{k+1} - x_k \rangle))\nonumber \\
    &= \eE\left(\langle \zeta_{k+1}, \zeta_{k+1} - \frac{1}{\lambda} \nabla f(x_k) - \nabla \tilde g_{\sigma}(x_{k+1} -  \zeta_{k+1}) \rangle) \right) \le 2 L_h^2 \sigma^2
\end{align}
while $\eE \| \zeta_{k+1} - \zeta_k\|^2\leq 2(\eE \| \zeta_{k+1}\|^2+\eE \| \zeta_{k}\|^2 )\leq  8 L_h^2 \sigma^2 $. 
By taking the expectation in equation~\eqref{eq:intermediar_bound} and thanks to equation~\eqref{eq:control_noise_intermediar} and the previous relation, we get 
\begin{align}
    &\eE(f(x_{k+1} - \zeta_{k+1})) \le \eE(f(x_k - \zeta_k)) + \left(L_f - \lambda\right) \eE(\|x_{k+1} - x_k\|^2) \nonumber
    \\&- \lambda \eE(\langle \nabla \tilde g_{\sigma}(x_{k+1} -  \zeta_{k+1}), x_{k+1} - x_k \rangle)   + (8 L_f + 2\lambda) L_h^2 \sigma^2 \nonumber\\
    &+ \eE\left(\langle \nabla f(x_k- \zeta_k) - \nabla f(x_k), x_{k+1} - x_k \rangle \right).\label{eq:intermediar_bound_1}
\end{align}

We have,
\begin{align}
    &\eE\left(\langle \nabla f(x_k- \zeta_k) - \nabla f(x_k), x_{k+1} - x_k \rangle \right) \nonumber \\
    &\le \eE\left( \|\nabla f(x_k- \zeta_k) - \nabla f(x_k)\| \|x_{k+1} - x_k\|\right)  \nonumber\\
    &\le \eE\left( L_f \| \zeta_k\| \|x_{k+1} - x_k\| \right)  \nonumber\\
    &\le \frac{L_f}{2} \eE\left(\| \zeta_k\|^2\right) + \frac{L_f}{2} \eE\left(\|x_{k+1} - x_k\|^2\right)  \nonumber\\
    &\le \frac{L_f}{2} \eE\left( \|x_{k+1} - x_k\|^2\right) + L_f L_h^2 \sigma^2, \label{eq:one_term_control}
\end{align}
using equation~\eqref{eq:variance_zeta} for the last inequality. By combining equations~\eqref{eq:intermediar_bound_1} and~\eqref{eq:one_term_control}, we get
\begin{align}
    &\eE(f(x_{k+1} - \zeta_{k+1})) \le \eE(f(x_k - \zeta_k)) + \left(\frac{3}{2} L_f - \lambda\right) \eE(\|x_{k+1} - x_k\|^2) \nonumber
    \\&- \lambda \eE(\langle \nabla \tilde g_{\sigma}(x_{k+1} -  \zeta_{k+1}), x_{k+1} - x_k \rangle)   + (9 L_f + 2\lambda) L_h^2 \sigma^2. \label{eq:intermediar_result_2}
\end{align}

Thanks to Proposition~\ref{prop:denoiser_is_a_prox}, we have that $\tilde g_{\sigma}$ is $\rho$-weakly convexity, with $\rho = \frac{L_h}{1+L_h}$. So $\tilde g_{\sigma} + \frac{\rho}{2} \|\cdot\|$ is convex. This implies that $\forall x, y \in \R^d$, $-\langle \nabla \tilde g_{\sigma}(x), x - y \rangle \le \tilde g_{\sigma}(y) - \tilde g_{\sigma}(x) + \frac{\rho}{2} \|x-y\|^2$ (see a proof of this fact in Lemma~\ref{lemma:inequality_weakly_cvx} of Section~\ref{sec:snopnp}).
Therefore, with $ x = x_{k+1} -  \zeta_{k+1}$ and $y = x_k - \zeta_k$ in the previous inequality, we get
\begin{align}
    &-\eE(\langle \nabla \tilde g_{\sigma}(x_{k+1} -  \zeta_{k+1}), x_{k+1} - x_k \rangle) \nonumber \\
    &= -\eE(\langle \nabla \tilde g_{\sigma}(x_{k+1} -  \zeta_{k+1}), x_{k+1} -  \zeta_{k+1} - x_k -  \zeta_{k} \rangle) \nonumber\\
    &\le \eE\left( \tilde g_{\sigma}(x_k - \zeta_k) - \tilde g_{\sigma}(x_{k+1} -  \zeta_{k+1}) + \frac{\rho}{2} \|x_{k+1} -x_k + \zeta_k -  \zeta_{k+1}\|^2 \right) \nonumber\\
    &\le \eE\left( \tilde g_{\sigma}(x_k - \zeta_k) - \tilde g_{\sigma}(x_{k+1} -  \zeta_{k+1}) + \rho \alpha \|x_{k+1} -x_k\|^2 + \frac{4 \rho}{\alpha} L_h^2 \sigma^2 \right),\label{eq:control_scalar_product_weakly_convex}
\end{align}
because $\forall x, y \in \R^d, \forall \alpha > 0, \langle x, y\rangle \le \frac{\alpha}{2} \|x\|^2 + \frac{1}{2\alpha}\|x\|^2$.

By injecting Equation~\eqref{eq:control_scalar_product_weakly_convex} into Equation~\eqref{eq:intermediar_result_2} and denoting $\tilde \fF = f + \lambda \tilde g_{\sigma}$, we get
\begin{align}
    &\eE(\tilde \fF(x_{k+1} - \zeta_{k+1})) \le \eE(\tilde \fF(x_k - \zeta_k)) + \left(\frac{3}{2} L_f - \lambda (1 - \rho \alpha) \right) \eE(\|x_{k+1} - x_k\|^2) \nonumber
    \\&
    + (9 L_f + \frac{4 \rho \lambda}{\alpha} + 2\lambda ) L_h^2 \sigma^2. \label{eq:intermediar_result_3}
\end{align}
 Due to the inequality $\rho < 1$, we can take $\alpha = \frac{1}{2}$ to ensure that $1 - \rho \alpha \ge \frac{1}{2}$. With this choice, we obtain
\begin{align}
    &\eE(\tilde \fF(x_{k+1} - \zeta_{k+1})) \le \eE(\tilde \fF(x_k - \zeta_k)) + \frac{3L_f - \lambda (2 - \rho)}{2} \eE(\|x_{k+1} - x_k\|^2) \nonumber
    \\& + (9 L_f + 8 \rho \lambda + 2 \lambda) L_h^2 \sigma^2, \label{eq:intermediar_result_4}
\end{align}

By re-arranging the terms, we have
\begin{align*}
    \frac{ \lambda (2 - \rho) - 3L_f}{2} \eE(\|x_{k+1} - x_k\|^2) &\le \eE(\tilde \fF(x_k - \zeta_k) - \tilde \fF(x_{k+1} - \zeta_{k+1})) \\
    &+ (9 L_f + 8 \rho \lambda + 2\lambda ) L_h^2 \sigma^2.
\end{align*}

By averaging for $k$ between $0$ and $N-1$, we get, for $\lambda > \frac{3 L_f}{2-\rho}$
\begin{align*}
  \frac{1}{N} \sum_{k=0}^{N-1} \eE(\|x_{k+1} - x_k\|^2) &\le \frac{2(\tilde \fF(x_0) - \fF^\ast)}{N(\lambda (2 - \rho) - 3L_f)} + \frac{2(9 L_f + 8 \rho \lambda + 2 \lambda)}{\lambda (2 - \rho) - 3L_f} L_D^2 \sigma^2,
\end{align*}
due to $\zeta_0 = 0$.
This proves Lemma~\ref{lemma:residual_control_perturbed_pgd} with $C_1 = \frac{2}{(\lambda (2 - \rho) - 3L_f)}$ and $C_2 = \frac{2(9 L_f + 8 \rho \lambda + 2 \lambda)}{\lambda (2 - \rho) - 3L_f} L_h^2$. Moreover because $\rho = \frac{L_h}{1+L_h}< 1$ (Assumption~\ref{ass:denoiser_structure}), the condition $\lambda > \frac{3 L_f}{2-\rho}$ can be relaxed into $\lambda \ge 3 L_f$.

\subsection{Proof of Lemma~\ref{lemma:residual_control_snopnp}}\label{sec:proof_residuals_snopnp}

For completeness, we recall the proof of Lemma~\ref{lemma:residual_control_snopnp} in our notations, first demonstrated in~\cite[Lemma 2]{renaud2024convergenceanalysisproximalstochastic}.

First, we recall that for a $\rho$-weakly convex function $g_{\sigma}$ satisfying  $\rho<1$, problem $\argmin_{z \in \R^d} \frac{1}{2}\|x-z\|^2+g_{\sigma}(z)$ is strongly convex and $\mathsf{Prox}_{g_{\sigma}}$ is univalued. Next we introduce the quantity $G_{k}$ from the  proximal mapping~\eqref{eq:snopnp_iterates} as
\begin{align}
    \tilde G_{k} = x_k - x_{k+1}= x_k - \mathsf{Prox}_{g_{\sigma}}\left( x_k - \frac{1}{\lambda} \tilde \nabla f(x_k) \right),\label{eq:def_G2}
\end{align}
where $\tilde \nabla f(x_k) = \nabla f(x_k) + \zeta_k = \nabla f(x_k) + \sigma \lambda z_{k+1}$ with $z_{k+1} \sim \mathcal{N}(0, I_d)$.

From Assumption~\ref{ass:regularities_assumptions}(a), we have that $\nabla f$ is $L_{f}$-Lipschitz, which gives
\begin{align}
    f(x_{k+1}) &\le f(x_{k}) + \langle \nabla f(x_{k}), x_{k+1} - x_k \rangle +\frac{L_{f}}{2} \|x_{k+1} - x_k \|^2 \\
    &= f(x_{k}) - \langle \nabla f(x_{k}), \tilde G_{k} \rangle +\frac{L_{f}}{2} \|\tilde G_{k}\|^2 \\
    &= f(x_{k}) - \langle \tilde \nabla f(x_{k}), \tilde G_{k} \rangle +\frac{L_{f}}{2} \|\tilde G_{k}\|^2 + \langle \zeta_k, \tilde G_{k} \rangle, \label{eq:first_inequalities2}
\end{align}
with $\zeta_k =\tilde \nabla f(x_{k}) - \nabla f(x_{k})$.

The optimal condition of the proximal operator in~\eqref{eq:def_G2} implies that 
\begin{align}
    x_{k+1}-x_k + \frac{1}{\lambda} \tilde \nabla f(x_{k}) + \nabla g_{\sigma}(x_{k+1}) = 0,
\end{align}
so $\tilde G_{k}$ can also be expressed as
\begin{align}\label{eq:prox_map_formula2}
    \tilde G_{k} = \frac{1}{\lambda} \tilde \nabla f(x_{k}) + \nabla g_{\sigma}(x_{k+1}).
\end{align}
By using equations~\eqref{eq:first_inequalities2} and~\eqref{eq:prox_map_formula2}, we get
\begin{align}
    &f(x_{k+1}) \le f(x_{k}) - \langle \lambda\left( \tilde G_{k} - \nabla g_{\sigma}(x_{k+1}) \right), \tilde G_{k} \rangle +\frac{L_{f}}{2} \|\tilde G_{k}\|^2 + \langle \zeta_k, \tilde G_{k} \rangle \\
    &= f(x_{k}) + \left(\frac{L_{f}}{2} - \lambda \right) \|\tilde G_{k}\|^2 + \lambda \langle \nabla g_{\sigma}(x_{k+1}) , \tilde G_{k} \rangle + \langle \zeta_k, \tilde G_{k} \rangle \\
    &= f(x_{k}) + \left(\frac{L_{f}}{2} - \lambda \right) \|\tilde G_{k}\|^2 + \lambda \langle \nabla g_{\sigma}(x_{k+1}) , x_k - x_{k+1} \rangle + \langle \zeta_k, \tilde G_{k}\rangle. \label{eq:inequality_before_cvx2}
\end{align}

Combining equation~\eqref{eq:inequality_before_cvx2} and Lemma~\ref{lemma:inequality_weakly_cvx}, we get
\begin{align*}
   & f(x_{k+1}) \\\le\, &f(x_{k}) + \left(\frac{L_{f}}{2} - \lambda \right) \|\tilde G_{k}\|^2 + \lambda \left( g_{\sigma}(x_k) - g_{\sigma}(x_{k+1}) + \frac{\rho}{2} \|x_{k+1}-x_k \|^2 \right)\\& + \langle \zeta_k, \tilde G_{k} \rangle \\
    =\, &f(x_{k}) + \left(\frac{L_{f}}{2} - \lambda \right) \|\tilde G_{k}\|^2 + \lambda \left( g_{\sigma}(x_k) - g_{\sigma}(x_{k+1}) + \frac{\rho}{2} \|\tilde G_{k}\|^2 \right) \\&+ \langle \zeta_k, \tilde G_{k} \rangle \\
    =\,& f(x_{k}) + \left(\frac{L_{f}}{2} + \frac{\rho \lambda}{2} - \lambda \right) \|\tilde G_{k}\|^2 + \lambda \left( g_{\sigma}(x_k) - g_{\sigma}(x_{k+1}) \right) + \langle \zeta_k, \tilde G_{k} \rangle.
\end{align*}
Then we introduce $G_{k}$ as
\begin{align}
    G_{k} = x_k - \mathsf{Prox}_{g_{\sigma}}\left( x_k - \frac{1}{\lambda} \nabla f(x_k) \right).\label{def_Gbar}
\end{align}
Note that $G_k$ is a deterministic function of $x_k$. Thus we have $\eE_k \langle \zeta_k, G_{k} \rangle = 0$.

By re-arranging terms we get
\begin{align}
   \left(\lambda (1 - \frac{\rho}{2}) - \frac{L_{f}}{2}  \right) \|\tilde G_{k}\|^2\nonumber \le F(x_{k}) - F(x_{k+1}) + \langle \zeta_k, G_{k} \rangle + \langle \zeta_k, \tilde G_{k} - G_{k} \rangle.
\end{align}
By taking the expectation $\eE_k$ with respect with $x_k$, and taking $\lambda > \frac{L_f}{2- \rho}$, we get
\begin{align}
\eE_k\left(\|\tilde G_{k}\|^2\right) \le \frac{2}{(2-\rho)\left(\lambda - \frac{L_f}{2-\rho} \right)} \left[\eE_k\left(F(x_{k}) - F(x_{k+1})\right) + \eE_k\left(\langle \zeta_k, \tilde G_{k} - G_{k} \rangle\right) \right].\label{eq:control_with_scalar_product}
\end{align}
To finish the proof, we now prove an intermediate result.
\begin{lemma}\label{lemma:control_prox_map_weakly}
For $f$  $\rho$-weakly convex, we have
\begin{align}
    \frac{1}{\lambda (1 - \rho)} \| \zeta_k \|  \ge \|\tilde G_{k} - G_{k}\|.
\end{align}
\end{lemma}
\begin{proof}
We define $u_k = \mathsf{Prox}_{g_{\sigma}}\left( x_k - \frac{1}{\lambda} \nabla f(x_k) \right)$ and $v_k = \mathsf{Prox}_{g_{\sigma}}\left( x_k - \frac{1}{\lambda} \tilde \nabla f(x_k) \right)$. By the optimal condition of the proximal operator, we get
\begin{align*}
   u_k - x_k + \frac{1}{\lambda} \nabla f(x_k) + \nabla g_{\sigma}(u_k) = 0 \\
   v_k - x_k + \frac{1}{\lambda} \tilde \nabla f(x_k) + \nabla g_{\sigma}(v_k) = 0.
\end{align*}
So we have
\begin{align*}
    \nabla f(x_k) &= - \lambda \nabla g_{\sigma}(u_k) + \lambda (x_k - u_k)\\
    \tilde \nabla f(x_k) &= - \lambda \nabla g_{\sigma}(v_k) + \lambda (x_k - v_k) \\
    \langle \nabla f(x_k), v_k - u_k \rangle &=  - \lambda \langle \nabla g_{\sigma}(u_k), v_k - u_k \rangle + \lambda \langle x_k - u_k, v_k - u_k \rangle \\
    \langle \tilde \nabla f(x_k), u_k - v_k \rangle &=  - \lambda  \langle \nabla g_{\sigma}(v_k), u_k - v_k \rangle + \lambda \langle x_k - v_k, u_k - v_k \rangle \\
    \langle \nabla f(x_k) - \tilde \nabla f(x_k), v_k - u_k \rangle &= \lambda \langle \nabla g_{\sigma}(v_k) - \nabla g_{\sigma}(u_k), v_k - u_k \rangle +\lambda \|v_k - u_k \|^2\\
    &\ge \lambda \left( 1 - \rho \right) \|v_k - u_k \|^2,
\end{align*}
where we used Lemma~\ref{lemma:inequality_weakly_cvx} on the $\rho$-weak convexity of $f$ for the last relation. 
Thus, from the definition of the bias $\zeta_k = \tilde \nabla f(x_k) - \nabla f(x_k)$, we obtain 
\begin{align*}
    \| \zeta_k \| \| v_k - u_k \| \ge \langle \nabla f(x_k) - \tilde \nabla f(x_k), v_k - u_k \rangle \ge \lambda \left( 1 - \rho \right) \|v_k - u_k \|^2.
\end{align*}
By definition of $G_{k}$ and $\tilde G_{k}$, we have $\tilde G_{k} - G_{k} = v_k - u_k$ and  
 we get
\begin{align*}
    \| \zeta_k \|  \ge \lambda \left( 1 - \rho \right) \| \tilde G_{k} - G_{k} \|.
\end{align*}
\end{proof}
Using Lemma~\ref{lemma:control_prox_map_weakly} on equation~\eqref{eq:control_with_scalar_product}, we obtain
\begin{align}
   & \eE_k\left(\|\tilde G_{k}\|^2\right) \le \frac{2}{(2-\rho)\left(\lambda - \frac{L_f}{2-\rho} \right)}  \left[\eE_k\left(F(x_{k}) - F(x_{k+1})\right) +  \frac{1}{\lambda (1 - \rho)} \eE_k\left(\| \zeta_k \|^2\right) \right].
\end{align}
Moreover, we recall that $\zeta_k = \sigma \lambda z_{k+1}$, so
\begin{align}
    \eE_k(\|\zeta_k\|^2) = \sigma^2 \lambda^2 \eE_k(\|z_{k+1}\|^2) = \sigma^2 \lambda^2 d.
\end{align}
Thus, we obtain
\begin{align}\label{eq:bound_gadient_mapping}
    \eE_k\left(\|\tilde G_{k}\|^2\right) \le \frac{2}{(2-\rho)\left(\lambda - \frac{L_f}{2-\rho} \right)}  \left[\eE_k\left(F(x_{k}) - F(x_{k+1})\right) +  \frac{\lambda  \sigma^2 d}{1 - \rho} \right].
\end{align}
Recalling that $\tilde G_{k} = x_k - x_{k+1}$ and 
averaging for k between $0$ and $N-1$, we get
\begin{align*}
\frac{1}{N}\sum_{k=0}^{N-1} \eE\left(\|x_{k+1} - x_k\|^2\right) &\le \frac{2}{(2-\rho)\left(\lambda - \frac{L_f}{2-\rho} \right)}  \left[\frac{1}{N}\eE\left(F(x_0) - F(x_{N+1})\right) +  \frac{\lambda  \sigma^2 d}{1 - \rho} \right] \\
&\le \frac{2}{(2-\rho)\left(\lambda - \frac{L_f}{2-\rho} \right)}  \left[\frac{F(x_0) - F^\ast}{N} +  \frac{\lambda d \sigma^2}{1 - \rho} \right]
\end{align*}
which demonstrates Lemma~\ref{lemma:residual_control_snopnp} with $A_1 = \frac{2}{(2-\rho)\left(\lambda - \frac{L_f}{2-\rho} \right)}$ and $A_2 = \frac{2 \lambda d}{(1-\rho)(2-\rho)\left(\lambda - \frac{L_f}{2-\rho} \right)}$.

\subsection{Proof of Proposition~\ref{prop:snopnp_convergence_critical_point}}\label{sec:appendix_proof_proposition_cvg_critical_points_spgd}

We estimate the difference between the gradient of $F=f+\lambda g_{\sigma}$ and the gradient mapping $G_k$ (defined in equation~\eqref{eq:prox_map_formula2})

\begin{align*}
    \frac{1}{\lambda} \nabla F(x_k) - G_k &= \frac{1}{\lambda} \left(\nabla f(x_k) + \lambda \nabla g_{\sigma}(x_k) \right) - \frac{1}{\lambda} \tilde \nabla f(x_k) - \nabla g_{\sigma}(x_{k+1})\\
    &= \frac{1}{\lambda} \left(\nabla f(x_k) - \tilde \nabla f(x_k) \right) + \nabla g_{\sigma}(x_k) - \nabla g_{\sigma}(x_{k+1}).
\end{align*}

Inspired by~\cite[Lemma 5.5]{gao2024non}, we make the following computation, using the $L_g$-Lipschitz continuity of $\nabla g_{\sigma}$ on $\text{Im}(D_{\sigma})$ given by Assumption~\ref{ass:g_l_smooth}
\begin{align*}
    &\|\frac{1}{\lambda}\nabla F(x_k)\|^2 = \|G_k + \frac{1}{\lambda} \nabla F(x_k) - G_k\|^2 \\
    &=  \|G_k + \frac{1}{\lambda} \left(\nabla f(x_k) - \tilde \nabla f(x_k)\right) + \left( \nabla g_{\sigma}(x_k) - \nabla g_{\sigma}(x_{k+1}) \right)\|^2 \\
    &\le 3 \|G_k\|^2 + \frac{3}{\lambda^2} \|\nabla f(x_k) - \tilde \nabla f(x_k)\|^2 + 3\| \nabla g_{\sigma}(x_k) - \nabla g_{\sigma}(x_{k+1})\|^2\\
    &\le 3 \|G_k\|^2 + \frac{3}{\lambda^2} \|\nabla f(x_k) - \tilde \nabla f(x_k)\|^2 + 3 L_g^2 \| x_k - x_{k+1}\|^2 \\
    &\le 3(1+L_g^2) \|G_k\|^2 + \frac{3}{\lambda^2} \|\nabla f(x_k) - \tilde \nabla f(x_k)\|^2.
\end{align*}

Then, taking the expectation with respect to $x_k$, we get
\begin{align}
    \eE_k\left(\|\nabla F(x_k)\|^2\right) &\le 3 \lambda^2 (1+L_g^2) \eE_k\left(\|G_k\|^2\right) + 3 \eE_k\left(\|\nabla f(x_k) - \tilde \nabla f(x_k)\|^2 \right) \\
    &\le 3 \lambda^2 (1+L_g^2) \eE_k\left(\|x_{k+1} - x_k\|^2\right) + 3 \sigma^2 \lambda^2,
\end{align}
where we use the definition of $\tilde \nabla f(x_k) = \nabla f(x_k) + \sigma \lambda z_{k+1}$, with $z_{k+1} \sim \mathcal{N}(0, I_d)$.

By taking the total expectation, averaging for $k$ between $0$ and $N-1$ and applying Lemma~\ref{lemma:residual_control_snopnp}, we obtain

\begin{align}
    \frac{1}{N} \sum_{k=0}^{N-1}\eE_k\left(\|\nabla F(x_k)\|^2\right)
    &\le  \frac{6 \lambda^2 (1+L_g^2)}{(2-\rho)\left(\lambda - \frac{L_f}{2-\rho} \right)}  \left[\frac{F(x_0) - F^\ast}{N} +  \frac{\lambda  \sigma^2}{1 - \rho} \right] + 3 \sigma^2 \lambda^2 \\
    &\le \frac{C_1 (F(x_0) - F^\ast)}{N} + C_2 \sigma^2, 
\end{align}
with $B_1 = \frac{6 \lambda^2 (1+L_g^2)}{(2-\rho)\left(\lambda - \frac{L_f}{2-\rho} \right)}$ and $B_2 = 3 \lambda^2 + \frac{6 \lambda^3 (1+L_g^2)}{(2-\rho)(1-\rho)\left(\lambda - \frac{L_f}{2-\rho} \right)}$.

\subsection{Proof of Corollary~\ref{cor:almost_sre_cvg_subsequence}}\label{sec:proof_cor_almost_sure_cvg_subsequence}
Equation~\eqref{eq:prop_snopnp_cvg_critical_point} shows that there exists a subsequence $x_{\psi_0(k)}$ such that $\forall k \in \N$ 
\begin{align}
    \eE\left(\|\nabla F(x_{\psi_0(k)})\|^2\right) \le 2 B_2 \sigma^2,
\end{align}
so that
\begin{align}
    \eE\left(\|\nabla F(x_{\psi_0(k)})\|\right) \le \sqrt{\eE\left(\|\nabla F(x_{\psi_0(k)})\|^2\right)} \le \sqrt{2 B_2} \sigma.
\end{align}
By the Fatou Lemma, we get
\begin{align}
    \eE\left(\liminf_{k \to + \infty}\|\nabla F(x_{\psi_0(k)})\|\right) \le \liminf_{k \to + \infty} \eE\left(\|\nabla F(x_{\psi_0(k)})\|\right) \le \sqrt{2 B_2} \sigma.
\end{align}
Then, by the Markov inequality, for $\beta > 0$, we get
\begin{align}
    \Prb{\liminf_{k \to + \infty}\|\nabla F(x_{\psi_0(k)})\| \ge \frac{\sqrt{2 B_2} \sigma}{\beta}} \le \beta.
\end{align}
So, with probability larger than $1 - \beta$, we have
\begin{align}\label{eq:prob_markov_control}
    \liminf_{k \to + \infty} \|\nabla F(x_{\psi_0(k)})\| < \frac{\sqrt{2 B_2} \sigma}{\beta}.
\end{align}
We define the event $B = \{ \lim_{k \to +\infty}\|x_{\psi_0(k)}\| = +\infty \}$. We get
\begin{align}
    \eE\left(\|\nabla F(x_{\psi_0(k)})\|^2\right) &= \eE\left(\|\nabla F(x_{\psi_0(k)})\|^2 \mathds{1}_B\right) + \eE\left(\|\nabla F(x_{\psi_0(k)})\|^2 \mathds{1}_{\bar{B}}\right) \\
    &\ge \eE\left(\|\nabla F(x_{\psi_0(k)})\|^2 \mathds{1}_B\right).
\end{align}
By the Fatou lemma, it leads to
\begin{align}
    2B_2 \sigma^2 \ge \liminf_{k \to +\infty}\eE\left(\|\nabla F(x_{\psi_0(k)})\|^2 \mathds{1}_B\right) \ge \eE\left(\liminf_{k \to +\infty}\|\nabla F(x_{\psi_0(k)})\|^2 \mathds{1}_B\right).
\end{align}
Since $\nabla F$ is coercive,  $\liminf_{k \to +\infty}\|\nabla F(x_{\psi_0(k)})\| = +\infty$.  on the event $B$. So, necessarily, $\Prb{B} = 0$. Then, almost surely, $x_{\psi_0(k)}$ does not go to $+\infty$. Hence, almost surely, there exists $x_{\psi_1(k)}$ a subsequence of $x_{\psi_0(k)}$ which is bounded in a compact $\mathbf{K} \subset \R^d$.

Then, by equation~\eqref{eq:prob_markov_control}, with probability larger than $1-\beta$, we have
\begin{align}\label{eq:prob_markov_control_2}
    \liminf_{k \to + \infty} \|\nabla F(x_{\psi_1(k)})\| \le \frac{\sqrt{2 B_2} \sigma}{\beta}.
\end{align}
Due to the inequality~\eqref{eq:prob_markov_control_2}, there exists a subsequence $x_{\psi_2(k)}$ such that
\begin{align}\label{eq:prob_markov_control_3}
    \|\nabla F(x_{\psi_2(k)})\| \le \frac{2\sqrt{2 B_2} \sigma}{\beta}.
\end{align}

By Remark 6.5 in~\cite{bierstone1988semianalytic}, because $(x_{\psi_2(k)})_{k \in \N}$ is bounded in $\cK$ and $\fF$ is subanalytic on the compact $\cK$, the Lojasiewicz's inequality ensures that there exist $r, c > 0$ (depending of $\cK$) such that $\forall x \in \cK$
\begin{align}\label{eq:lojasiewicz_inequality}
    \| \nabla F(x) \| \ge  c d(x, \sS)^{\frac{1}{r}},
\end{align}
with $d(x, \sS) = \min_{s \in \sS} \|x-s\|$ the distance to the set of critical points of $F$. Note that $s, r >0$ depend of the realization. Therefore with probability larger than $1-\beta$, we have
\begin{align}
    \| \nabla F(x_{\psi_2(k)}) \| \ge  c d(x_{\psi_2(k)}, \sS)^{\frac{1}{r}}.
\end{align}

By combining the previous inequality and equation~\eqref{eq:prob_markov_control_3}, we get that, with a probability larger than $1-\beta$, 
\begin{align}
    c d(x_{\psi_2(k)}, \sS)^{\frac{1}{r}} &\le \| \nabla F(x_{\psi_2(k)}) \| \le \frac{2\sqrt{2 B_2} \sigma}{\beta} \\
    d(x_{\psi_2(k)}, \sS)  &\le \left(\frac{2\sqrt{2 B_2}}{c}\right)^r \left(\frac{\sigma}{\beta}\right)^r,
\end{align}
that proves Corollary~\ref{cor:almost_sre_cvg_subsequence} with $\psi = \psi_2$, $B_3 = \left(\frac{2\sqrt{2 B_2}}{c}\right)^r$.

\end{document}